\documentclass[%
reprint,
nofootinbib,
mathtools,
amssymb,
aps,
prl,
superscriptaddress,
usenames,dvipsnames,
floatfix
]{revtex4-2}

\makeatletter
\newcommand{\customlabel}[2]{%
	#1\protected@edef\@currentlabel{#1}\label{#2}%
}
\makeatother

\usepackage{amssymb, amsmath, xcolor, enumitem, amsthm, scrextend, setspace, hyperref}
\usepackage[scr]{rsfso}
\hypersetup{colorlinks=true,
	breaklinks=true,
	linkcolor=[rgb]{0, 0.2, 0.5},
	citecolor=[rgb]{0, 0.35, 0},      
	urlcolor=[rgb]{0, 0.2, 0.5},
	pdfpagemode=UseOutlines}
\usepackage[most]{tcolorbox}
\tcbuselibrary{skins}
\usepackage{cleveref}

\crefname{inequality}{Ineq.}{Ineqs.}
\crefname{figure}{Fig.}{Figs.}
\crefname{equation}{Eq.}{Eqs.}
\newtheorem{theorem}{Theorem}
\newtheorem{lemma}[theorem]{Lemma}

\newtheorem{fact}[theorem]{Fact}
\newtheoremstyle{myplain}{3pt}{0pt}{\upshape}{}{\bfseries} {.}{.5em}{} 
\theoremstyle{myplain}
\newtheorem{defn}[theorem]{Definition}

\usepackage{tikz}
\usetikzlibrary{quantikz2}
\usepackage[caption=false]{subfig}
\usepackage{adjustbox}
\usepackage{algpseudocode}
\usepackage{algorithm}

\definecolor{blue(ncs)}{rgb}{0.0, 0.53, 0.74}
\definecolor{ao(english)}{rgb}{0.0, 0.5, 0.0} 
\definecolor{darklavender}{rgb}{0.45, 0.31, 0.59}
\definecolor{burntornge}{rgb}{0.8, 0.33, 0.0} 
\definecolor{carnelian}{rgb}{0.7, 0.11, 0.11} 

\newcommand{\ugrshort}{Quantum Thermodynamics and Computation Group, Electromagnetism and Matter Physics Department, University of Granada, Granada, Spain}
\newcommand{\jku}{Department for Quantum Information and Computation at Kepler (QUICK), Johannes Kepler University, Linz, Austria}
\newcommand{\uottawa}{Department of Mathematics and Statistics, University of Ottawa, Ottawa, Canada}
\newcommand{\luh}{Institut f{\"{u}}r Theoretische Physik, Leibniz Universit{\"{a}}t Hannover, Appelstra{\ss}e 2, 30167 Hannover, Germany}
\newcommand{\bilkent}{Department of Mathematics, Bilkent University, Ankara, Turkey}

\begin{document}
	
	\title{Blind Quantum Computation with a Small Quantum Server}
	
	\author{Daniel Lovsted}
	\affiliation{\uottawa}
	
	\author{Filipa C.\,R. Peres}
	\affiliation{\ugrshort}
	\affiliation{\jku}
	
	\author{Joshua Nevin}
	\affiliation{\uottawa}
	
	\author{Selman Ipek}
	\affiliation{\luh}
	\affiliation{\bilkent}
	
	\author{Anne Broadbent}
	\affiliation{\uottawa}

	\begin{abstract}
		Blind quantum computation (BQC) allows low-resource clients to securely delegate computations to a quantum server, but server resource costs scale with the computation size, posing a bottleneck for implementations. By leveraging Pauli-based computation (PBC), we achieve BQC with a server whose size depends only on the non-Clifford gate count. Our protocol inherits fault tolerance and qubit virtualization from PBC and reveals that classical-client BQC is possible even in the absence of classical simulability. Finally, we present an entanglement-based dual protocol that performs a resource state computation with a dramatically reduced execution cost.
	\end{abstract}
	
	\maketitle
	
	\textit{Introduction.---}Quantum computation promises substantial speedups over classical computation for important classes of problems~\cite{Pre18}. For the foreseeable future, quantum computing resources are likely to be scarce and computations will thus be outsourced to remote quantum servers, leaving clients' data vulnerable to theft by a malicious server. The cryptographic task of blind quantum computation (BQC) addresses this risk: In BQC, a low-resource client delegates a computation to a quantum server while keeping aspects of it secret.
	
	Many BQC protocols have been proposed~\cite{Chi05, BFK09, Bro15}, differing primarily in their resource costs for the client, communication, and server. Client resources are well studied and have been improved to levels believed to be optimal or near-optimal for protocols with information-theoretic security~\cite{ACGK19}. The same is true of communication resources~\cite{MPDF13, DK25arxiv}. Server resources, however, remain largely unoptimized, posing an obstacle to continued progress in implementing BQC. While recent experimental demonstrations have succeeded in interfacing photonic client-side devices with trapped-ion~\cite {DNM+24} and solid-state~\cite{WSS+25} servers, opening the door for BQC servers to harness advances in matter-based platforms~\cite{DMP+21, AWR+22, PDF+21, MBA+23}, it also raises the question of how BQC protocols can most efficiently use server-side processors that are still scale-limited.
	
	In this Letter, we present a BQC protocol that significantly reduces the quantum server size, while matching the best-known client and communication resource costs: For an $n$-qubit quantum circuit $C$ with $t$ non-Clifford~$T$ gates, we require a quantum server with only $t$ qubits, compared to previous methods whose server size scales with both $n$ \textit{and} $t$. The protocol allows the client to hide a classical input $\textbf{x}$ with perfect security and receive a sample from the output distribution of $C\ket{\textbf{x}}$, while requiring just four-way random state preparation and a single step of quantum communication of $t$ qubits, matching the best client and communication costs of previous input-hiding protocols~\cite{Bro15}.
	
	Our protocol extends the Pauli-based computation~(PBC) model of~\cite{BSS16}, which samples from the output distribution of any universal quantum computation via a sequence of adaptively chosen Pauli measurements performed on a magic state. We introduce blindness to PBC by having the client apply secret, random rotations to the initial state, request Pauli measurements from the server via a modified classical algorithm, and perform a final correction to the measurement outcomes returned by the server. Our protocol inherits other benefits of PBC, including fault tolerance and the possibility of classically ``virtualizing'' $k$ more qubits than are physically available at a cost exponential in $k$ only, while maintaining security. Moreover, we identify a class of circuits which is believed \textit{not} to be efficiently classically simulable and which our protocol executes with perfect security via a fully classical client. This makes progress on an open question in quantum cryptography about the limits of classical clients in BQC~\cite[Sec.~5]{BS16}.
	
	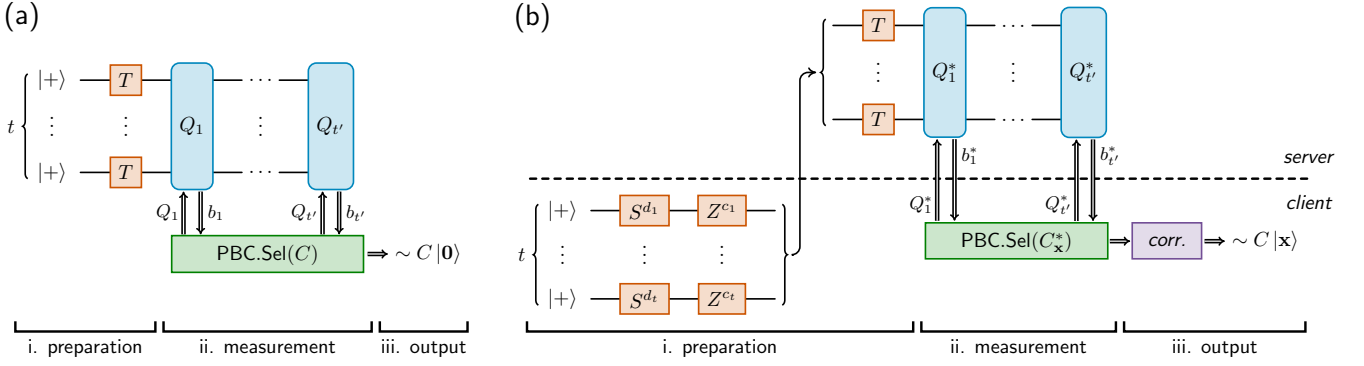
\begin{figure*}[t]
		\begin{adjustbox}{width=\textwidth}
			\tikzset{double_arrow/.style={-{Implies[]}, thick, double}
			}
			
			
			\begin{tikzpicture}
				\node(C){
					\begin{quantikz}[ampersand replacement=\&,wire types={q,n,q}]
						\lstick[3,braces={transform canvas={xshift=5pt}},label style={xshift=5pt}]{$t$} \midstick{\ket{+}} \& \gate[style={fill=burntornge!20,draw=burntornge}]{T} \& \gate[3, style={fill=blue(ncs)!20,draw=blue(ncs),rounded corners}]{Q_1} \& \push{~\dots~} \& \gate[3, style={fill=blue(ncs)!20,draw=blue(ncs),rounded corners}]{Q_{t^\prime}} \\
						\midstick[label style={yshift=3pt}]{\vdots} \& \midstick[label style={yshift=3pt}]{\vdots} \& \& \midstick[label style={yshift=3pt}]{\vdots} \& \\
						\midstick{\ket{+}} \& \gate[style={fill=burntornge!20,draw=burntornge}]{T} \& \& \push{~\dots~} \&
					\end{quantikz}
				};
				
				\draw[double_arrow]
				([xshift=-2pt,yshift=-1pt]\tikzcdmatrixname-3-3.south east) -- ([xshift=-2pt,yshift=-22pt]\tikzcdmatrixname-3-3.south east) node[midway,xshift=-4pt,label={[font=\small]right:{$b_1$}}] {};
				
				\draw[double_arrow]
				([xshift=-3pt,yshift=-1pt]\tikzcdmatrixname-3-5.south east) -- ([xshift=-3pt,yshift=-22pt]\tikzcdmatrixname-3-5.south east)
				node[midway,xshift=-4pt,label={[font=\small]right:{$b_{t^\prime}$}}] {};
				
				\draw[double_arrow]
				{([xshift=10pt,yshift=-22pt]\tikzcdmatrixname-3-3.south west) -- node[midway,xshift=5.5pt,label={[font=\small]left:{$Q_1$}}] {} ([xshift=10pt,yshift=-1pt]\tikzcdmatrixname-3-3.south west)} ;
				
				\draw[double_arrow]
				{([xshift=11pt,yshift=-22pt]\tikzcdmatrixname-3-5.south west) -- node[midway,xshift=5.5pt,label={[font=\small]left:{$Q_{t^\prime}$}}] {} ([xshift=11pt,yshift=-1pt]\tikzcdmatrixname-3-5.south west)} ;
				
				\draw[ao(english), fill=ao(english)!20, thick] ([xshift=-6pt,yshift=-31pt]\tikzcdmatrixname-3-3) rectangle ([xshift=20pt,yshift=-49pt]\tikzcdmatrixname-3-5) node[midway, yshift=-0.5pt, style={black}] {$\mathsf{PBC.Sel}(C)$};
				
				\draw[double_arrow]
				{([xshift=22pt,yshift=-40pt]\tikzcdmatrixname-3-5.center) -- ([xshift=32pt,yshift=-40pt]\tikzcdmatrixname-3-5.center)} node[xshift=-4pt,label = right:{$\sim C\ket{\textbf{0}}$}] {};
				
				\draw[very thick] {([xshift=-14pt,yshift=-73pt]\tikzcdmatrixname-3-1.center) -- ([xshift=-14pt,yshift=-78pt]\tikzcdmatrixname-3-1.center) -- ([xshift=-14pt,yshift=-78pt]\tikzcdmatrixname-3-3.center) node [midway, below] {\textsf{i. preparation}} -- ([xshift=-14pt,yshift=-73pt]\tikzcdmatrixname-3-3.center)} {};
				
				\draw[very thick] {([xshift=-10pt,yshift=-73pt]\tikzcdmatrixname-3-3.center) -- ([xshift=-10pt,yshift=-78pt]\tikzcdmatrixname-3-3.center) -- ([xshift=24pt,yshift=-78pt]\tikzcdmatrixname-3-5.center) node [midway, below] {\textsf{ii. measurement}} -- ([xshift=24pt,yshift=-73pt]\tikzcdmatrixname-3-5.center)} {};
				
				\draw[very thick] {([xshift=28pt,yshift=-73pt]\tikzcdmatrixname-3-5.center) -- ([xshift=28pt,yshift=-78pt]\tikzcdmatrixname-3-5.center) -- ([xshift=71pt,yshift=-78pt]\tikzcdmatrixname-3-5.center) node [midway, below] {\textsf{iii. output}} -- ([xshift=71pt,yshift=-73pt]\tikzcdmatrixname-3-5.center)} {};
				
				\node[font=\Large] at ([xshift=14pt,yshift=16pt]C.north west) {\textsf{(a)}};
				
			\end{tikzpicture}
			
			\hspace{0.3cm}
			
			
			\begin{tikzpicture}
				\tikzset{double_arrow/.style={-{Implies[]}, thick, double}
				}
				\node(C){
					\begin{quantikz}[ampersand replacement=\&,wire types={n,n,n,n,q,n,q}]
						\& \& \& \& \& \lstick[3]{} \& \gate[style={fill=burntornge!20,draw=burntornge}]{T}\setwiretype{q} \& \gate[3, style={fill=blue(ncs)!20,draw=blue(ncs),rounded corners}]{Q^*_1} \& \push{~\dots~} \& \gate[3, style={fill=blue(ncs)!20,draw=blue(ncs),rounded corners}]{Q^*_{t^\prime}} \\
						\& \& \& \& \& \& \midstick[label style={yshift=3pt}]{\vdots} \& \& \midstick[label style={yshift=3pt}]{\vdots} \& \\
						\& \& \& \& \& \& \gate[style={fill=burntornge!20,draw=burntornge}]{T}\setwiretype{q} \& \& \push{~\dots~} \& \\
						\& \& \& \& \& \& \& \& \& \& \& \& \& \& \& \& \\
						\lstick[3,braces={transform canvas={xshift=5pt}},label style={xshift=5pt}]{$t$} \midstick{\ket{+}} \& \gate[disable auto height, style={fill=burntornge!20,draw=burntornge}][24pt]{S^{d_1}} \& \gate[style={fill=burntornge!20,draw=burntornge}]{Z^{c_1}} \& \rstick[3]{} \& \setwiretype{n} \& \& \& \& \& \& \& \&  \\
						\midstick[label style={yshift=3pt}]{\vdots} \& \midstick[label style={yshift=3pt}]{\vdots} \& \midstick[label style={yshift=3pt}]{\vdots} \& \& \& \& \& \& \& \& \&  \\
						\midstick{\ket{+}} \& \gate[disable auto height, style={fill=burntornge!20,draw=burntornge}][24pt]{S^{d_t}} \& \gate[style={fill=burntornge!20,draw=burntornge}]{Z^{c_t}} \& \& \setwiretype{n} \& \& \& \& \& \& \& \& \& \& \& \& 
					\end{quantikz}
				};
				
				\draw[line cap=round, dashed, very thick]{([xshift=-12pt,yshift=-6pt]\tikzcdmatrixname-4-1.center) -- ([xshift=4pt,yshift=-6pt]\tikzcdmatrixname-4-17.center)} node [label = above:{\textsf{\textit{server}}}] [label = below:{\textsf{\textit{client}}}] {};
				
				\draw[thick, rounded corners, -{>[length=4pt, width=5pt]}] {([xshift=-1pt]\tikzcdmatrixname-6-5.center) -- ([xshift=2pt]\tikzcdmatrixname-6-5.center) -- ([xshift=2pt]\tikzcdmatrixname-2-5.center) -- ([xshift=10pt]\tikzcdmatrixname-2-5.center)} {};
				
				\draw[double_arrow]
				{([xshift=-1.5pt,yshift=-1pt]\tikzcdmatrixname-3-8.south east) -- node[midway,xshift=-4pt,yshift=12pt,label={[font=\small]right:{$b^*_{1}$}}] {}([xshift=-1.5pt,yshift=-41pt]\tikzcdmatrixname-3-8.south east)} ;
				
				\draw[double_arrow]
				{([xshift=-2.5pt,yshift=-1pt]\tikzcdmatrixname-3-10.south east) -- node[midway,xshift=-4pt,yshift=12pt,label={[font=\small]right:{$b^*_{t^\prime}$}}] {} ([xshift=-2.5pt,yshift=-41pt]\tikzcdmatrixname-3-10.south east)} ;
				
				\draw[double_arrow]
				{([xshift=10.5pt,yshift=-41pt]\tikzcdmatrixname-3-8.south west) -- node[midway,xshift=5.5pt,yshift=-12pt,label={[font=\small]left:{$Q^*_{1}$}}] {}([xshift=10.5pt,yshift=-1pt]\tikzcdmatrixname-3-8.south west)} {};
				
				\draw[double_arrow]
				{([xshift=11.5pt,yshift=-41pt]\tikzcdmatrixname-3-10.south west) -- node[midway,xshift=5.5pt,yshift=-12pt,label={[font=\small]left:{$Q^*_{t^\prime}$}}] {}([xshift=11.5pt,yshift=-1pt]\tikzcdmatrixname-3-10.south west)} {};
				
				\draw[ao(english), fill=ao(english)!20, thick] ([xshift=-5.5pt,yshift=-6pt]\tikzcdmatrixname-5-8) rectangle ([xshift=-9.5pt,yshift=-2pt]\tikzcdmatrixname-6-11) node[midway, yshift=-0.5pt, style={black}] {$\mathsf{PBC.Sel}(C^*_\textbf{x})$};
				
				\draw[darklavender, fill=darklavender!20, thick] ([xshift=2.5pt,yshift=-6pt]\tikzcdmatrixname-5-11) rectangle ([xshift=36.5pt,yshift=-2pt]\tikzcdmatrixname-6-11) node[midway, yshift=-0.5pt,style={black}] {\textsf{\textit{corr.}}};
				
				\draw[double_arrow]
				{([xshift=-8.5pt,yshift=-15pt]\tikzcdmatrixname-5-11.center) -- ([xshift=1.5pt,yshift=-15pt]\tikzcdmatrixname-5-11.center)} {};
				
				\draw[double_arrow]
				{([xshift=9.5pt,yshift=-15pt]\tikzcdmatrixname-5-13.center) -- ([xshift=19.5pt,yshift=-15pt]\tikzcdmatrixname-5-13.center)} node[xshift=-4pt,label = right:{$\sim C\ket{\textbf{x}}$}] {};
				
				\draw[very thick] {([xshift=-13.5pt,yshift=-11pt]\tikzcdmatrixname-7-1.center) -- ([xshift=-13.5pt,yshift=-16pt]\tikzcdmatrixname-7-1.center) -- ([xshift=-11.5pt,yshift=-16pt]\tikzcdmatrixname-7-8.center) node [midway, below] {\textsf{i. preparation}} -- ([xshift=-11.5pt,yshift=-11pt]\tikzcdmatrixname-7-8.center)} {};
				
				\draw[very thick] {([xshift=-7.5pt,yshift=-11pt]\tikzcdmatrixname-7-8.center) -- ([xshift=-7.5pt,yshift=-16pt]\tikzcdmatrixname-7-8.center) -- ([xshift=-5.5pt,yshift=-16pt]\tikzcdmatrixname-7-11.center) node [midway, below] {\textsf{ii. measurement}} -- ([xshift=-5.5pt,yshift=-11pt]\tikzcdmatrixname-7-11.center)} {};
				
				\draw[very thick] {([xshift=-1.5pt,yshift=-11pt]\tikzcdmatrixname-7-11.center) -- ([xshift=-1.5pt,yshift=-16pt]\tikzcdmatrixname-7-11.center) -- ([xshift=30.5pt,yshift=-16pt]\tikzcdmatrixname-7-15.center) node [midway, below] {\textsf{iii. output}} -- ([xshift=30.5pt,yshift=-11pt]\tikzcdmatrixname-7-15.center)} {};
				
				\node[font=\Large] at ([xshift=14pt,yshift=-12pt]C.north west) {\textsf{(b)}};
			\end{tikzpicture}
		\end{adjustbox}
		
		\caption{(a) Pauli-based computation (PBC) and (b) Blind PBC, for a circuit~$C$ with~$t$~$T$ gates. Information flows from left to right; single (resp.~double) lines denote quantum (resp.~classical) information. Each protocol has three stages. (i) Preparation of an initial state. In Blind PBC, the client chooses secret random bits~$c_i,d_i \in \{0,1\}$ and applies rotations. (ii) Measurement of $t^\prime\leq t$ Pauli observables~$Q_j$, with outcomes $b_j \in \{0,1\}$, selected adaptively by the classical $\textsf{PBC.Sel}$ algorithm. In Blind PBC, the client runs~$\textsf{PBC.Sel}$ on a circuit~$C^*_\textbf{x}$ derived from~$C$ and her secret classical input~$\textbf{x}$. (iii) Output of a sample from the desired distribution, determined by~$\textsf{PBC.Sel}$. In Blind PBC, the client corrects~$\textsf{PBC.Sel}$ classically before returning the output.}
		
		\label{fig:pbc}
	\end{figure*}
	
	Finally, we present an entanglement-based dual to our blind protocol in which the Pauli measurement sequence is performed \textit{prior} to the choice of input. Under this time reversal, the dual property to the cryptographic input-blindness property is input-independence: We show that the dual protocol prepares a computation-dependent but input-independent resource state for the circuit~$C$, which can then be executed with minimal quantum resources to sample from $C\ket{\textbf{x}}$, for any~$\textbf{x}$. Specifically, for $C$ with $n$ qubits, depth $d$, and $t$ non-Clifford $T$ gates, the execution phase requires a~$t$-qubit quantum memory and the ability to perform single-qubit measurements in the Pauli-$X$ and $Y$ bases only. In allowing execution with a minimal set of measurement bases, our method resembles hypergraph-based computation~\cite{MM16, GGM19, TMH19}, but with a greatly reduced state size compared to the universal resource states of hypergraph models, whose size scales as $O(nd)$. Our model also parallels work on computation-specific measurement-based quantum computation (MBQC) using graph states~\cite{RBB03, VPG+24, KH25, PG25}, but unlike the latter it requires no non-Pauli measurements for execution. We thus position PBC as a bridge between hypergraph computation and computation-specific MBQC, capturing the resource optimality of each approach.

	\textit{Blind PBC versus existing BQC.---}We introduce \textit{Blind PBC}, a BQC protocol which allows a resource-limited client, with help from a quantum server, to sample from the output of a circuit~$C$ run on a computational basis input~$\ket{\textbf{x}}$ of her choosing, while maintaining perfect, information-theoretic secrecy of~$\textbf{x} \in \{0,1\}^n$ from the server. Blind PBC thus compares to BQC protocols that hide the input with perfect security in classical input/output regimes, among which the quantum computing on encrypted data (QCED) protocol~\cite{FBS+14, Bro15} has the lowest resource costs. Blind PBC improves on the server resource cost of QCED while matching client and communication costs: Blind PBC employs a server with just~$t$ qubits, while QCED uses a server of size~$n+t$ when run with a single instance of quantum communication. The client resources are identical in Blind PBC and QCED (four-way random state preparation of~$t$ qubits), as are communication costs ($t$ qubits, sent in one ahead-of-time, one-way instance). We note that QCED may also be executed with~$t$ separate instances of just-in-time quantum communication from client to server, in which case a server of~$n+1$ qubits suffices. Thus, Blind PBC offers a regime-specific server size improvement and a strict communication cost improvement. 
	
	Blind PBC inherits other benefits from PBC. It can be implemented fault-tolerantly, since it uses measurements of Pauli observables only. Additionally, in perfect analogy with PBC, it enables quantum memory reduction via the simulation of ``virtual qubits,'' so that the outcome of a computation on~$t+k$ qubits can be estimated on a register of size~$t$, with a sampling complexity exponential in~$k$ only~\cite{BSS16, PG23}. In the blind setting, client and communication resource costs also scale exponentially in $k$; moreover, the procedure remains perfectly secure.
	
	While QCED defaults to a classical client exactly when~$t=0$, i.e., for Clifford circuits, Blind PBC does so for a strictly larger class of circuits. This class includes circuits with arbitrarily many~$T$ gates and which are believed to display quantum advantage. The limits of classical-client BQC are an open question in quantum cryptography: While circuit-hiding, perfectly secure BQC of universal quantum computation is unlikely with a classical client~\cite{ACGK19}, little is known about the boundary between Cliffordness and universality, where a classical client becomes insufficient. Our result provides evidence that classical simulability is \textit{not} a necessary condition for a circuit class to be delegated by a classical client with perfect security.

	\textit{Pauli-based computation.---}Blind PBC is designed in the Pauli-based computation (PBC) model of~\cite{BSS16}. PBC produces a sample from the output distribution of an arbitrary quantum computation, expressed as a Clifford+$T$ circuit~$C$ on~$n$ qubits, with~$t$~$T$~gates, and~$w \leq n$ final readout measurements, run on the all-zero state~$\ket{\textbf{0}}$. It does so by an adaptive sequence of at most~$t$ nondestructive measurements of independent and commuting Pauli operators on~$t$ qubits.
	
	For comparison with the blind version, we present PBC in three stages (\cref{fig:pbc}(a)). We use~$\ket{+}$ for the~$+1$ eigenstate of the Pauli-$X$ operator, $T \coloneqq \text{diag}(1,e^{i\pi/4})$, and~$\ket{T} \coloneqq T\ket{+}$. In the first stage, the magic state~$\ket{T}^{\otimes t}$ is initialized. Second, a classical algorithm, \textsf{PBC.Sel}, takes as input~$C$ and previous measurement outcomes~$\{b_j\}_{j=1}^{k}$ to select the subsequent measurement $Q_{k+1}$. Third, \textsf{PBC.Sel} returns a string~$\textbf{y} \in \{0,1\}^w$, which constitutes the desired sample from~$C\ket{\textbf{0}}$. Some bits of $\textbf{y}$ are previous measurement outcomes~$b_j$ and others are classically determined.

	\textit{Blind PBC.---}We now present our BQC protocol (\cref{fig:pbc}(b)). $C$ has parameters~$n$, $t$, and~$w$ as above.\footnote{The protocol was first reported in the M.Sc. thesis of one of the authors \cite{Lov26}.}
	
	Blind PBC has three stages, each one a modification of the corresponding stage in regular PBC. First, the client samples random bits~$c_i,d_i \in \{0,1\}$, for~$i\in[t]$, where~$[t]$ stands for $\{1,...,t\}$. According to these values, the client applies the rotation~$S^{c_i}Z^{d_i}$ to qubit~$i$ of a register initialized in~$\ket{+}^{\otimes t}$, where~$S \coloneqq \text{diag}(1,i)$. These~$t$ qubits are passed to the server who performs a~$T$ operation on each one. Second, the client chooses her input~$\textbf{x}$ and classically produces a circuit~$C_\textbf{x}^*$ from~$C$, $\textbf{x}$, and the~$c_i, d_i$ values chosen previously. The client then runs \textsf{PBC.Sel} on~$C_\textbf{x}^*$ and sends a classical description of each Pauli operator~$Q^*_j$ it selects to the server, who measures the requested operator and replies with outcome~$b^*_j$, which the client passes to \textsf{PBC.Sel}. Third, once all measurements are complete and \textsf{PBC.Sel} returns~$\textbf{y}$, the client computes a classical correction vector,~$\textbf{u}^* \in \{0,1\}^w$, and returns~$\textbf{y} \oplus \textbf{u}^*$, which is a sample from the distribution of~$C\ket{\textbf{x}}$. For simplicity, our presentation considers computational basis states as input, but more generally, our protocol can hide an input state drawn from any orthonormal basis of stabilizer states.
	
	\Cref{alg:bpbc} is the procedure by which the client computes~$C^*_\textbf{x}$ and~$\textbf{u}^*$. For the correctness of Blind PBC, we will ultimately need it to hold that
	\begin{equation}
		\label{eq:corr}
		\mathsf{PBC}\left(\bigotimes_{i=1}^tS^{c_i}Z^{d_i}\ket{T}, C_\textbf{x}^*\right) \oplus \textbf{u}^* \sim C\ket{\textbf{x}}\,,
	\end{equation}
	where, for any state $\ket{\psi}$ and circuit~$A$, $\textsf{PBC}(\ket{\psi},A)$ denotes the probabilistic map from $\ket{\psi}$ and~$A$ to an output $\{0,1\}^w$ induced by measuring~$\ket{\psi}$ as specified by $\textsf{PBC.Sel}(A)$, and $\sim$ denotes equivalent output distributions.
	
	Conceptually, the algorithm ensures \cref{eq:corr} is satisfied as follows. To begin, the client represents~$\textbf{x}$ as a circuit-initial sequence of Pauli-$X$ gates. The resulting circuit~$C_{\textbf{x}}$ is such that $C_{\textbf{x}}\ket{\textbf{0}} \sim C\ket{\textbf{x}}$, straightforwardly. Next, the client propagates the initial Pauli operators to the end of the circuit, pushing them past each gate~$G_k$. If~$G_k$ is Clifford, the client computes an updated Pauli that exactly preserves the equivalence of the circuit; this is possible since the Clifford group normalizes the Pauli group. If $G_k = T$, the client first updates the Pauli and conditionally adds an~$S$ gate to preserve equivalence (see \cref{fig:gadg}(a)). Second, since Blind PBC uses magic states of the form $S^{d}Z^{c}\ket{T}$ (instead of $\ket{T}$, which $\textsf{PBC.Sel}$ expects), the client adds gates~$S^d$, $Z^{d\overline{b}\oplus c}$ after~$G_k$, where~$\overline{b}$ is the bitflip of~$b$, according to the equivalence shown in \cref{fig:gadg}(b) and similarly to \cite{FBS+14, Bro15}. Third, gates following~$G_k$ are combined to yield at most one~$S$ gate and a Pauli. (If two~$S$ gates are present, from the procedures in \cref{fig:gadg}(a) \textit{and}~(b), they are combined to~$Z^{ud}$, which joins the Pauli term.) The leftover~$S$ gate, with exponent labeled~$s_i$ at the~$i$-th~$T$ gate, is added to the circuit at that position, and the Pauli continues being propagated forward. The circuit obtained once the Paulis reach the end of the circuit, just before the readout measurements, is denoted~$C_\textbf{x}^\prime$.  We have
	\begin{equation}
		\label{eq:intermediate}
		\mathsf{PBC}\left(\bigotimes_{i=1}^tS^{c_i}Z^{d_i}\ket{T}, C_\textbf{x}^\prime\right) \sim C\ket{\textbf{x}}\,,
	\end{equation}
	since all steps during the propagation have either preserved equivalence compared to~$C_\textbf{x}\ket{\textbf{0}}$ or have accounted for the rotations on the magic states.
	
	Finally, the Pauli operators are removed from the end of the circuit. The resulting circuit is $C_{\textbf{x}}^*$. The effect of the removed Paulis on the measurement outcomes would have been to flip the readout whenever the measured wire held an~$X$ gate. Thus, recording the indices of $X$ components as the correction vector $\textbf{u}^*$, deleting the Paulis from the circuit, and adding~$\textbf{u}^*$ to the circuit outcome, preserves the equivalence. Therefore, \cref{eq:corr} holds. 
	
	\begin{algorithm}[H]
		\caption{Obtaining~$C^*_\textbf{x}$ and~$\textbf{u}^*$ in Blind PBC. For $P \in \{X,Z\}$ and $\textbf{r} \in \{0,1\}^n$, $P(\textbf{r})$ is the $n$-qubit Pauli with~$P$ on qubit~$j$ whenever $\textbf{r}_j = 1$ and~$I$ otherwise. Vectors $\textbf{u}, \textbf{v} \in \{0,1\}^n$ track Pauli $X(\textbf{u})Z(\textbf{v})$. Vector positions are indicated by a subscript index (e.g. $\textbf{v}_\ell$).}
		\label{alg:bpbc}
		\textbf{Input:}  $n$-qubit circuit $C = (G_1...G_m,M)$ for gates~$G_k$ and~$M\subseteq [n]$ readout measurement indices with $|M| = w$; input~$\textbf{x}\in\{0,1\}^n$; and~$c_i,d_i \in \{0,1\}$ sampled before.\\
		\textbf{Output:} $C^*_{\textbf{x}}$ and correction vector~$\textbf{u}^* \in \{0,1\}^w$.
		\begin{algorithmic}[1]
			\State{$\textbf{u} \gets \textbf{x}$, $\textbf{v} \gets \textbf{0}$} \Comment{circuit-initial Pauli-$X$ sequence}
			\State{$C^*_{\textbf{x}} \gets \emptyset$} \Comment{initialize empty circuit}
			\State{$i \gets 0$} \Comment{$T$ gate counter}
			\For{$k=1$ to $m$} \Comment{push input to end}
			\If{$G_k \in \{H, S, \text{CNOT}\}$} \Comment{Clifford gate}
			\State{$\textbf{u} \gets \textbf{u}^\prime$ where $X({\textbf{u}^\prime}) = G_kX(\textbf{u}) G_k^\dagger$}
			\State{$\textbf{v} \gets \textbf{v}^\prime$ where $Z({\textbf{v}^\prime}) = G_kZ(\textbf{v}) G_k^\dagger$}
			\State{$C^*_{\textbf{x}} \gets C^*_{\textbf{x}}G_k$}
			\Else \Comment{$T$ gate}
			\State{$i \gets i+1$}
			\State{$\ell \gets$ index ($1$ to $n$) of wire that~$G_k$ acts on}
			\State{$\textbf{v}_\ell \gets \textbf{v}_\ell \oplus \textbf{u}_\ell \oplus \textbf{u}_\ell d_i \oplus d_i \overline{b_i} \oplus c_i $}
			\State{$s_i \gets \textbf{u}_\ell \oplus d_i$}
			\State{$C^*_{\textbf{x}} \gets C^*_{\textbf{x}}G_kS^{s_i}$}
			\EndIf
			\EndFor
			\State{$\textbf{u}^* \gets \textbf{u}|_M$, i.e. $\textbf{u}$ restricted to indices in~$M$}
			\State{\Return{$C^*_{\textbf{x}}$, $\textbf{u}^*$}}
		\end{algorithmic}
	\end{algorithm} 
	
	The security of Blind PBC is due to the form of~$C_\textbf{x}^*$. $C_\textbf{x}^*$ is the only part of the client-server interaction that depends on~$\textbf{x}$, and $C_\textbf{x}^*$ can differ from~$C$ only in that $C_\textbf{x}^*$ may have some inserted~$S$ gates, each immediately after a~$T$ gate. So, the server can learn about~$\textbf{x}$ at most what is contained in the exponents of the possible~$S$ gates, $s_i$ for~$i \in [t]$. But, informally, since each~$s_i$ has the form~$\textbf{u}_\ell \oplus d_i$, where $d_i$ is random and secret, the input-dependent content of $C_\textbf{x}^*$ is encrypted in a way comparable to the classical one-time pad. In the Supplemental Material, we define information-theoretic security for Blind PBC via a cryptographic simulator and formalize this intuition to prove that Blind PBC satisfies the security definition.
	
	The class of circuits for which Blind PBC can be executed by a classical client can now be characterized as those circuits for which~$s_i=0$ for all~$i \in [t]$ and all choices of input~$\textbf{x}$. In this case, since~$C = C^*_\textbf{x}$ for any~$\textbf{x}$, the client can request the outcome of a regular PBC of~$C$ on $\ket{T}^{\otimes t}$, effectively choosing~$c_i = d_i =0$, and use a final classical correction alone to sample from~$C\ket{\textbf{x}}$; this preserves correctness, and since no input-dependent communication passes from client to server, it is trivially secure. This class includes Clifford circuits (vacuously, since $t=0$) but is strictly larger. For example, it includes a subclass of instantaneous quantum polynomial-time (IQP) circuits with a diagonal layer composed of~$T$ and $CS := \text{diag}(1,1,1,i)$ gates. These circuits are believed to be hard to simulate classically and considered promising candidates for experimental quantum advantage demonstrations \cite{BMS16average, BMS17,CvdW25}.

	
	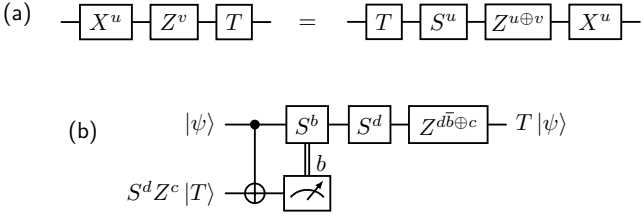
\begin{figure}[t]
		\begin{adjustbox}{width=\linewidth}
			\begin{tikzpicture}
				\node(C) {
					\begin{quantikz}[ampersand replacement=\&, column sep = 8pt]
						\& \gate{X^u} \& \gate{Z^v} \& \gate{T} \& \\        
					\end{quantikz}
				};
				\node at ([xshift=12pt,yshift=8pt]C.east) {=};
				
				\node at ([xshift=-12pt,yshift=-8pt]C.north west) {\textsf{(a)}};
				
			\end{tikzpicture}
			\hfill
			\begin{tikzpicture}
				\node(C) {
					\begin{quantikz}[ampersand replacement=\&, column sep = 8pt]
						\& \gate{T} \& \gate{S^u} \& \gate[disable auto height][28pt][15pt]{Z^{u \oplus v}} \& \gate{X^u} \& \\        
					\end{quantikz}
				};
			\end{tikzpicture}
		\end{adjustbox}
		
		\begin{adjustbox}{width=0.8\linewidth}
			\begin{tikzpicture}
				\node(B) {
					\begin{quantikz}[ampersand replacement=\&, column sep=8pt]
						\lstick{$\ket{\psi}$} \& \ctrl{1} \& \gate[disable auto height][18pt][16pt]{S^{b}} \& \gate[disable auto height][18pt][16pt]{S^d} \& \gate[disable auto height, label style={yshift=0.5pt}][34pt][16pt]{Z^{d \overline{b}\oplus c}} \& \rstick{$T\ket{\psi}$} \\
						\lstick{$S^dZ^c\ket{T}$} \& \targ{} \& \meter[label style={xshift=6pt,yshift=-2pt}]{b} \wire[u][1]{c}      
				\end{quantikz} };
				\node at ([xshift=-12pt,yshift=-18pt]B.north west) {\textsf{(b)}};
			\end{tikzpicture}
		\end{adjustbox}
		
		\caption{Circuit identities used to derive $C_\textbf{x}^*$ in Blind PBC. (a)~Propagating a Pauli past a~$T$ gate introduces a conditional~$S$ correction. (b)~Gadget for performing a~$T$ operation, where correction gates cancel out rotations on the magic state.}
		\label{fig:gadg}
	\end{figure}

	\textit{Tailored resource computation.---}Now we consider an entanglement-based, time-reordered dual protocol to Blind PBC, in which the Pauli measurements occur prior to the choice of input (\cref{fig:resource}). The dual protocol, which we call a tailored resource computation (TRC), produces a resource state on $t$ qubits specific to a computation represented as a circuit~$C$. The resource state is input-independent: It can be measured to produce a sample from $C\ket{\textbf{x}}$ for \textit{any} $\textbf{x} \in \{0,1\}^n$, chosen after the resource is prepared. The size of the resource state is exactly~$t$ qubits and the execution stage uses only single-qubit Pauli-$X$ and $Y$ measurements.
	
	Like previous works on computation-specific MBQC~\cite{RBB03, VPG+24, KH25, PG25}, we find that quantum resources scale exactly with non-Clifford circuit elements in the computation-specific setting. However, previous computation-specific work uses graph state resources, meaning that non-stabilizer measurements are required to compute. Our method loads non-stabilizerness into the resource, so that only Pauli measurements are needed, thus reducing resources at the execution step.
	
	\begin{figure}[t]
		\begin{adjustbox}{width=\linewidth}
			\tikzset{double_arrow/.style={-{Implies[]}, thick, double}
			}
			
			\begin{tikzpicture}
				\node(C){
					\begin{quantikz}[ampersand replacement=\&,wire types={q,n,q,n,n,q,q}]
						\& \gate[style={fill=burntornge!20,draw=burntornge}]{T} \& \gate[3, style={fill=blue(ncs)!20, draw=blue(ncs), rounded corners}]{Q_1} \& \push{~\dots~} \& \gate[3, style={fill=blue(ncs)!20, draw=blue(ncs), rounded corners}]{Q_{t^\prime}} \& \setwiretype{n} \& \& \& \\
						\& \midstick[label style={yshift=3pt}]{\vdots} \& \& \midstick[label style={yshift=3pt}]{\vdots} \& \& \& \& \& \\
						\& \gate[style={fill=burntornge!20,draw=burntornge}]{T} \& \& \push{~\dots~} \& \& \setwiretype{n} \& \& \& \& \\
						\& \midstick[label style={yshift=-49pt}]{\vdots} \& \& \& \& \& \& \& \& \\
						\& \& \& \& \& \& \& \& \\
						\& \& \& \& \& \& \gate[{style={fill=blue(ncs)!20,draw=blue(ncs),rounded corners}}]{W_1} \& \setwiretype{n} \midstick[label style={xshift=-3pt,yshift=-11pt}]{$\ddots$} \& \& \\
						\& \& \& \& \& \& \& \& \gate[{style={fill=blue(ncs)!20,draw=blue(ncs),rounded corners}}]{W_t} 
					\end{quantikz}
				};
				
				\draw[line cap=round, dashed, very thick]{([xshift=-12pt,yshift=-12pt]\tikzcdmatrixname-5-1.center) -- ([xshift=-6pt,yshift=-12pt]\tikzcdmatrixname-5-6.center) -- ([xshift=-8pt,yshift=10pt]\tikzcdmatrixname-3-7.center) -- ([xshift=14pt,yshift=10pt]\tikzcdmatrixname-3-10.center)} node [label = {[xshift=-22pt]above:{\textsf{\textit{preparation space}}}}] [label = {[xshift=-22pt]below:{\textsf{\textit{execution space}}}}] {};
				
				\draw[ao(english), fill=ao(english)!20, thick] ([xshift=-8pt,yshift=-22pt]\tikzcdmatrixname-3-3) rectangle ([xshift=13pt,yshift=-40pt]\tikzcdmatrixname-3-5) node[midway, yshift=-0.5pt, style={black}] {$\mathsf{PBC.Sel}(C)$};
				
				\draw[carnelian, fill=carnelian!20, thick] ([xshift=-9pt,yshift=-22pt]\tikzcdmatrixname-3-7) rectangle ([xshift=12pt,yshift=-40pt]\tikzcdmatrixname-3-9) node[midway, yshift=-0.5pt, style={black}] {$\mathsf{Execute}(\textbf{x})$};
				
				\draw[double_arrow]
				([xshift=-4pt,yshift=0pt]\tikzcdmatrixname-3-3.south east) -- ([xshift=-4pt,yshift=-13pt]\tikzcdmatrixname-3-3.south east) {};
				
				\draw[double_arrow]
				([xshift=-5pt,yshift=0pt]\tikzcdmatrixname-3-5.south east) -- ([xshift=-5pt,yshift=-13pt]\tikzcdmatrixname-3-5.south east) {};
				
				\draw[double_arrow]
				{([xshift=8pt,yshift=-13pt]\tikzcdmatrixname-3-3.south west) -- ([xshift=8pt,yshift=0pt]\tikzcdmatrixname-3-3.south west)} {};
				
				\draw[double_arrow]
				{([xshift=9pt,yshift=-13pt]\tikzcdmatrixname-3-5.south west) -- ([xshift=9pt,yshift=0pt]\tikzcdmatrixname-3-5.south west)} {};
				
				\draw[double_arrow]
				([xshift=-2pt,yshift=-4pt]\tikzcdmatrixname-5-7.south east) -- ([xshift=-2pt,yshift=-14pt]\tikzcdmatrixname-5-7.south east) {};
				
				\draw[double_arrow]
				([xshift=-2pt,yshift=18pt]\tikzcdmatrixname-6-9.south east) -- ([xshift=-2pt,yshift=-21pt]\tikzcdmatrixname-6-9.south east) {};
				
				\draw[double_arrow]
				{([xshift=6pt,yshift=-14pt]\tikzcdmatrixname-5-7.south west) -- ([xshift=6pt,yshift=-4pt]\tikzcdmatrixname-5-7.south west)} {};
				
				\draw[double_arrow]
				{([xshift=6pt,yshift=-21pt]\tikzcdmatrixname-6-9.south west) -- ([xshift=6pt,yshift=18pt]\tikzcdmatrixname-6-9.south west)} {};
				
				\draw[double_arrow] {([xshift=14pt,yshift=-8pt]\tikzcdmatrixname-4-5.center) -- ([xshift=41pt,yshift=-8pt]\tikzcdmatrixname-4-5.center)} {};
				
				\draw[double_arrow]
				{([xshift=14pt,yshift=-30pt]\tikzcdmatrixname-3-9.center) -- ([xshift=24pt,yshift=-30pt]\tikzcdmatrixname-3-9.center)} node[xshift=-4pt,label = right:{$\sim C\ket{\textbf{x}}$}] {};
				
				\draw[line cap=round, thick]{([xshift=1.5pt,yshift=0pt] \tikzcdmatrixname-1-1.center) -- ([xshift=-10pt,yshift=14pt] \tikzcdmatrixname-4-1.center) -- ([xshift=1.5pt,yshift=0pt] \tikzcdmatrixname-6-1.center)};
				
				\draw[line cap=round, thick]{([xshift=1.5pt,yshift=0pt] \tikzcdmatrixname-3-1.center) -- ([xshift=-10pt,yshift=-6pt] \tikzcdmatrixname-5-1.center) -- ([xshift=1.5pt,yshift=0pt] \tikzcdmatrixname-7-1.center)};
				
				\draw[very thick] {([xshift=-7.5pt,yshift=-11pt]\tikzcdmatrixname-7-1.center) -- ([xshift=-7.5pt,yshift=-16pt]\tikzcdmatrixname-7-1.center) -- ([xshift=-2pt,yshift=-16pt]\tikzcdmatrixname-7-6.center) node [midway, below] {\textsf{i. preparation}} -- ([xshift=-2pt,yshift=-11pt]\tikzcdmatrixname-7-6.center)} {};
				
				\draw[very thick] {([xshift=2pt,yshift=-11pt]\tikzcdmatrixname-7-6.center) -- ([xshift=2pt,yshift=-16pt]\tikzcdmatrixname-7-6.center) -- ([xshift=56pt,yshift=-16pt]\tikzcdmatrixname-7-9.center) node [midway, below] {\textsf{ii. execution}} -- ([xshift=56pt,yshift=-11pt]\tikzcdmatrixname-7-9.center)} {};
				
				\node at ([xshift=8pt,yshift=-50pt]C.north west) {
					\begin{quantikz}[ampersand replacement=\&,wire types={n,n,n}, row sep=21pt]
						\lstick[3]{$t$} \& \& \& \\
						\& \& \& \\
						\& \& \& \\
				\end{quantikz} };
				
				\node at ([xshift=8pt,yshift=-66pt]C.west) {
					\begin{quantikz}[ampersand replacement=\&,wire types={n,n}, row sep=26pt]
						\lstick[2]{$t$} \& \& \& \\
						\& \& \& \\
				\end{quantikz} };
				
			\end{tikzpicture}
		\end{adjustbox}
		
		\caption{Tailored resource computation (TRC). (i) Preparation: PBC measurements for $C\ket{\textbf{0}}$ are performed on~$t$ half-EPR pairs rotated by~$T$. The resulting state on the other $t$-qubit system is a computation-specific, input-independent resource state for $C$. (ii) Execution: The resource may be measured adaptively to sample from $C\ket{\textbf{x}}$ for any $\textbf{x} \in \{0,1\}^n$. Only single-qubit measurements in two Pauli bases are required for execution: for all~$i$, $W_i \in \{X,Y\}$.
		}
		\label{fig:resource}
	\end{figure}
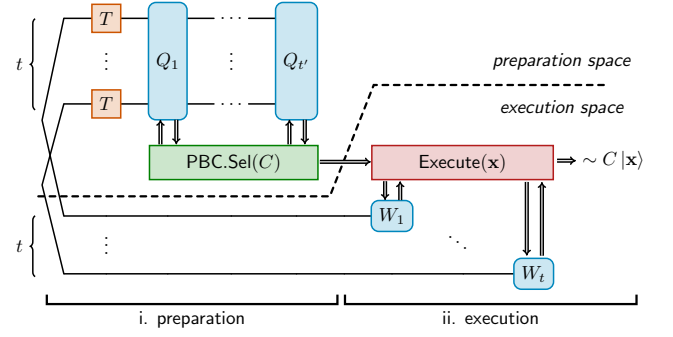

	Our work also relates to hypergraph computation models~\cite{MM16, GGM19, TMH19} and other generalizations of graph state computation~\cite{KvdW19}, where non-Cliffordness is similarly located in the resource state to simplify execution measurements to Pauli bases only. To our knowledge, however, such models have not been considered in a computation-specific context. Instead, they aim for universality, and consequently produce large resource states which scale in~$n$ and the circuit depth~$d$, usually with substantial overhead; for example, one model which (like ours) requires only two Pauli measurement bases for execution employs a resource state of $d(2n+63)\binom{n}{3}-n$ qubits~\cite{TMH19}. We match the simplicity of their measurement set while requiring a greatly reduced resource state of size of just~$t$.
	
	TRC thus bridges a theoretical gap between computation-specific graph state models and universal hypergraph models. We recover the best of both worlds, capturing the resource size reductions of the former and the measurement simplicity of the latter, thus easing the resource cost for the execution step. Because of this simplicity and growing interest in computation-specific MBQC in general~\cite{FDB+21}, we believe that, in addition to its conceptual interest, our model has significant practical implications.
	
	We show the correctness of TRC by a derivation from Blind PBC, inspired by the security proof of \cite{Bro15}. First, in the preparation stage of Blind PBC, suppose the client prepares~$t$ EPR pairs, sends one half of each to the server, draws random bits $d_i$, $i \in [t]$, measures the other half in the~$X$ (resp.~$Y$) basis if $d_i=0$ (resp.~$d_i=1$), and records the outcome as $c_i$. The state on each of the server's half-pairs is now exactly $S^{d_i}Z^{c_i}\ket{+}$, so this preparation is equivalent to that of Blind PBC, and if the client continues as in Blind PBC, the outcome will be a sample from~$C\ket{\textbf{x}}$. Next, suppose the client delays the measurements on her half-pairs. That is, she sends generic half-EPR pairs to the server, who performs a~$T$ rotation on each, and measurements are selected according to~$\mathsf{PBC.Sel}(C)$. Since $C = C_\textbf{x}^*$ when $s_1 = ... = s_t = 0$, this can be seen as the client setting all~$s_i$ values to~$0$. Then, after the measurements are complete, the client chooses~$\textbf{x}$ freely and computes~$d_i$ values to satisfy $s_i = \textbf{u}_\ell \oplus d_i$ according to the~$\textbf{u}$ given by her $\textbf{x}$. The client performs single-qubit measurements, $X$ or~$Y$ on qubit~$i$ (if $d_i = 0$ or $1$, respectively), recording each outcome as~$c_i$. Note that the client must compute the~$d_i$ values in order, from~$i=1$ to~$n$, since~$d_i$ may depend on~$c_j$ for $j<i$. This procedure is identical to the previous one, except that the client's measurements are delayed, which is valid since measurements on separate subsystems commute. If the final correction is performed as in Blind PBC, the outcome is therefore again a sample from~$C\ket{\textbf{x}}$. We have now arrived at TRC from Blind PBC; we conclude by dividing the procedure into preparation and execution stages and tasks as shown in \cref{fig:resource}.

	\textit{Outlook.---}Our Blind PBC protocol advances BQC towards resource-optimality by reducing the size of the server's quantum processor. Conceptually, our results show that non-Clifford circuit elements alone pose the relevant limit on register size in the blind setting, independently of input size. Our protocol for tailored resource computation (TRC) exploits a duality between input-blindness and input-independence, and can be understood as pushing the limits of precomputation of a quantum circuit, in that it minimizes the resources needed for its execution. We may ask: Can these resources be reduced further, and where do their fundamental limits lie? On a practical level, Blind PBC may enable near-term experimental demonstration of BQC at greater scale, while TRC reduces the resource barriers for the execution of pre-chosen quantum circuits. Our work thus has implications for the future of real-world quantum computing, suggesting that in a world of constrained quantum resources, useful applications may be available sooner than anticipated.
	
	\emph{Acknowledgments.---}
	A.B., D.L., and J.N. acknowledge the support of the Natural Sciences and Engineering Research Council of Canada (NSERC) (ALLRP 569582-21), and of the Canada Research Chairs Program (CRC-2023-00173).
	F.C.R.P. acknowledges support from Ayuda Consolidación CNS2023-145392 (MICIU\slash AEI\slash 10.13039\slash 501100011033, NextGenerationEU\slash PRTR) and the European Research Council (ERC) via the Starting grant q-shadows (101117138). S.I. acknowledges support from Nieders{\"{a}}chsisches Ministerium f{\"{u}}r Wissenschaft und Kultur. All authors acknowledge funding from EU HORIZON RIA FoQaCiA GA 101070558.

	\newpage
	
	\clearpage
	
	\onecolumngrid
	\begin{center}
		\noindent\textbf{Supplemental Material for: \emph{Blind Quantum Computation with a Small Quantum Server}}
	\end{center}
	
	\setcounter{equation}{0}
	\setcounter{figure}{0}
	\setcounter{table}{0}
	\setcounter{theorem}{0}
	\makeatletter
	\renewcommand{\thetable}{S\arabic{table}}
	\renewcommand{\theequation}{S\arabic{equation}}
	\renewcommand{\thefigure}{S\arabic{figure}}
	\renewcommand{\bibnumfmt}[1]{[S#1]}
	\renewcommand{\thetheorem}{S\arabic{theorem}}
	\renewcommand{\thelemma}{S\arabic{lemma}}
	\renewcommand*{\thesection}{\Alph{section}}
	\renewcommand*{\thesubsection}{\arabic{subsection}}
	\renewcommand*{\thesubsubsection}{\alph{subsubsection}}

	\twocolumngrid

	Here, we formalize and prove the security of the main protocol Blind PBC. Our security requirements for Blind PBC are formalized using a \emph{simulator-based} definition of security, in which the server attacks some idealized (and easy-to-study) protocol with the property that any attack against the real protocol can also be conducted against the idealized protocol (see \cite{L17} for an overview of the simulation proof technique). 
	
	\begin{defn}\label{FormalSimulatorDefn} Let $\mathbf{P}$ denote an interactive quantum protocol to be executed between a client and a quantum server on input a classical bit string. We define an \emph{initialization for $\mathbf{P}$} to be a pair $(X, Y)$ where $X$ and $Y$ are algorithms corresponding respectively to the client and the server of the protocol (with $Y$ possibly deviating from the protocol), where $X$ and $Y$ have access to respective Hilbert spaces $\mathcal{A}_X$ and $\mathcal{B}_Y$\footnote{We use the letters $\mathcal{A}$ and $\mathcal{B}$ to emphasize the difference between the client and server, often respectively referred to as Alice and Bob in cryptographic settings. In particular, the client has access to the quantum registers $\mathcal{A}_X$ but does not necessarily have any quantum abilities beyond the ability to prepare and send single-qubit states, whereas the server can entangle his registers with those of $\mathcal{A}_X$.} and $\delta\in\{0,1\}^*$ is an initial input accepted by the protocol. Given such a $\delta$, the tuple $(X,Y,\delta)$ induces a quantum channel $\Phi_{\delta}:L(\mathcal{A}_X\otimes\mathcal{B}_Y)\mapsto L(\mathcal{Z})$, where $\mathcal{Z}$ is the output space of the channel. Furthermore, given any quantum computer $Y$ with input registers $\mathcal{B}_Y$ and output registers $\mathcal{Z}$, we define a $Y$-\emph{simulator} $\mathscr{S}_{Y}$ to be an algorithm which takes as input a $\delta\in\{0,1\}^*$ and outputs a quantum channel $\Sigma_{\delta}:L(\mathcal{B}_Y)\mapsto L(\mathcal{Z})$ (where this channel depends on $\delta$). We call $\Sigma_{\delta}$ the channel \emph{induced} by $\mathscr{S}_{Y}$ on input $\delta$. \end{defn}

	\Cref{FormalSimulatorDefn} allows us to make the notion of privacy precise. Closely adhering to the security definition of~\cite{Bro15}, we define the following:
	
	\begin{defn}\label{EpsilonSecurity} Let $\mathbf{P}$ be a delegated quantum protocol. We say that $\mathbf{P}$ is \emph{perfectly secure} if, for any $\mathbf{P}$-initialization $(X,Y)$ (with $Y$ possibly deviating) there exists a $Y$-simulator $\mathscr{S}_{Y}$ such that, for every classical input $\delta\in\{0,1\}^*$ accepted as input by $\mathbf{P}$, we have $||\Phi_{\delta}-\Sigma'_{\delta}||_{\diamond}=0$, where $||\cdot||_{\diamond}$ denotes the diamond norm and:
		\begin{enumerate}[label=\arabic*)]
			\item $\Phi_{\delta}:L(\mathcal{A}_X\otimes\mathcal{B}_Y)\mapsto L(\mathcal{Z})$ is as in \Cref{FormalSimulatorDefn}; AND
			\item $\Sigma'_{\delta}:L(\mathcal{A}_X\otimes\mathcal{B}_Y)\mapsto L(\mathcal{Z})$ is a channel consisting of first tracing out the $\mathcal{A}_X$-system and then applying the channel $\Sigma_{\delta}$ induced by $\mathscr{S}_{Y}$ on input $\delta$.
	\end{enumerate} \end{defn}
	
	Informally, \Cref{EpsilonSecurity} states that the server can learn no more (via his interaction with the client) about the client's hidden input than if he never had any interaction with the client at all. Thus, our remaining task is to prove the following.
	
	\begin{theorem}\label{MainSecurityThm} Blind PBC is perfectly secure. \end{theorem}
	
	To prove \Cref{MainSecurityThm}, we first introduce an auxiliary form of Blind PBC. The auxiliary protocol, denoted \ref{Prot:HiddenCompBasisAux}, is specified in the box below. The difference between Blind PBC and \ref{Prot:HiddenCompBasisAux} occurs after the client has already generated the classical circuit $C_{\mathbf{x}}^*$ obtained in \Cref{alg:bpbc} in the main text. In the auxiliary protocol, the client sends the server some extra classical information at each step after this initial conversion. Informally, in \ref{Prot:HiddenCompBasisAux}, the client is revealing to the server the initial adaptive circuit piece-by-piece, in a way that (as we prove below) does not reveal $\mathbf{x}$. In Blind PBC, at each step, the client computes the next Pauli and tells the server to measure in that basis, whereas in \ref{Prot:HiddenCompBasisAux}, the client and server each locally compute the Pauli determining the next measurement basis. 
	
	\begin{tcolorbox}[breakable, before upper={\setlength{\parskip}{0.5em}}, title=\customlabel{Auxiliary BPBC}{Prot:HiddenCompBasisAux}: (Auxiliary Form of Blind PBC for the proof of security), fonttitle=\bfseries, colbacktitle=teal!70!gray]
		
		\textbf{Input:} A pair $(C,\mathbf{x})$, where $C$ is a nonadaptive quantum circuit with $n$ qubits, $t$ $T$-gates, and $w$ readout measurements, and $\mathbf{x}\in\{0,1\}^n$ is a bit string which the client wishes to keep secret.
		
		\textbf{Output:} A sample from the classical output distribution of $C\lvert\mathbf{x}\rangle$ over $\{0,1\}^w$.
		
		\textbf{Protocol Description:}
		\vspace*{-2mm}
		\begin{enumerate}[label=\arabic*)]
			\item As in Blind PBC, the client draws a sequence of $2t$ bits $(\mathbf{c}, \mathbf{d})=(c_1, \cdots, c_t, d_1, \cdots, d_t)$ uniformly at random and prepares initial $t$-qubit density matrix $$\rho_0=\bigotimes_{i=1}^t S^{c_i}Z^{d_i}|T\rangle\langle T|Z^{d_i}(S^{c_i})^{\dagger}$$
			\item The client applies \Cref{alg:bpbc} on input $C,\mathbf{x}, \mathbf{c}, \mathbf{d}$ to produce classical circuit $C^*_{\mathbf{x}}$. 
			\item\label{enum:AuxProt3} We perform the following modified version of the $t$-step procedure specified in stage ii. of \Cref{fig:pbc}(b) (with $\rho_0$ being our initial input as in stage i. of \Cref{fig:pbc}(b)). In particular, during step $k$, where $0\leq k<t$, the client sends to the server a single bit $s_{k+1}$, and so, at the start of step $k$, the server has access to $s_1, \cdots, s_k$: Suppose we are at the start of step $0\leq k<t$. Each of the client and server can now locally determine the Pauli $P$ specifying the next measurement basis. In particular, the client, who has access to $C_{\mathbf{x}}^*$, can do this by applying $\textsf{PBC.Sel}$ to $C_{\mathbf{x}}^*$ and her previously obtained measurement outcomes, whereas the server can obtain $Q_{k+1}$ by using $s_1, \cdots, s_k$ as specified below (note that the server does not need anything in order to obtain the next Pauli in the $k=0$ step). The server performs the corresponding nondestructive measurement and sends the outcome $b_{k+1}$ to the client. Once the server and client both have $b_{k+1}$, the client uses this to compute the exponent $s_{k+1}\in\{0,1\}$ on the $S$ gate that comes immediately after the $(k+1)$th $T$-gate in time. The client then sends this value $s_{k+1}$ to the server.
			\item When all $t$ steps of \ref{enum:AuxProt3} are done,  the interaction with the server is finished and the client has a length-$(t+w)$ bit string whose final $w$ bits she corrects with $\mathbf{u}^*$ as obtained in \Cref{alg:bpbc}.
		\end{enumerate}
		
	\end{tcolorbox}
	
	For $0\leq k<t$, we can define a $(t+k)$-qubit density matrix which encodes all the information that the client has sent to the server up to the start of step $k$ of \ref{enum:AuxProt3} of \ref{Prot:HiddenCompBasisAux}, namely, the density matrix $\rho_k(\mathbf{c}, \mathbf{d}\vert \mathbf{s})$ defined as
	\begin{equation}\label{rhokfromcds}\left(\bigotimes_{i=1}^tS^{c_i}Z^{d_i}|T\rangle\langle T|Z^{d_i}\left(S^{c_i}\right)^{\dagger}\right)\otimes\left(\bigotimes_{i=1}^k|s_i\rangle\langle s_i|\right)\end{equation}
	where $\mathbf{s}\coloneqq(s_1, \cdots, s_t)$. Note that the separable $\rho_k(\mathbf{c}, \mathbf{d}\vert\mathbf{s})$ should be distinguished from the $t$-qubit density matrix $\tilde{\rho}_k$ that the server has obtained by starting with $\rho_0(\mathbf{c}, \mathbf{d})$ and performing the sequence of measurements specified by the protocol (where $\tilde{\rho}_k$, in general, corresponds to some arbitrary distribution of entangled $t$-qubit states). The purpose of the density matrix $\rho_k\left(\mathbf{c}, \mathbf{d}\vert\mathbf{s}\right)$ is to capture all of the quantum and classical information that the client has sent to the server so far (from which the server, if he chooses to follow the protocol honestly, can indeed produce $\tilde{\rho}_k$ by performing the prescribed sequence of quantum measurements and then tracing out the last $k$ qubits). From the client's perspective, Blind PBC is at least as secure as \ref{Prot:HiddenCompBasisAux}, since we are only giving the server more information than is passed to \ref{Prot:HiddenCompBasisAux}. We prove this formally below. 
	
	\begin{lemma}\label{lemma:AuxToMainPerfSec} If \ref{Prot:HiddenCompBasisAux} is perfectly secure, then Blind PBC is also perfectly secure. \end{lemma}
	
	\begin{proof} Suppose \ref{Prot:HiddenCompBasisAux} is perfectly secure. For any initial shared input circuit $C$ to \ref{Prot:HiddenCompBasisAux} (or to the main protocol Blind PBC, which accepts the same input shared between client and server), we write $\eta(C)$ to mean a bit string of length $\textnormal{poly}(|C|)$ encoding $C$. Now, we define a \emph{auxiliary-main translator} (or just \emph{AM translator}) to be a polynomial-time classical machine $T^{am}$ such that:
		
		\begin{enumerate}[label=T\arabic*)]
			\item\label{defT1} $T^{am}$ accepts as input any bit string of the form $\eta=\eta(C)$; AND
			\item\label{defT2} Given any such $\eta$ and any initialization $(X^m, Y^m)$ for Blind PBC and any initialization $(X^a, Y^a)$ for \ref{Prot:HiddenCompBasisAux}\footnote{The superscripts $m,a$ denote \emph{main} and \emph{auxiliary} respectively.}, the following hold:
			\begin{enumerate}[label=\roman*)]
				\item $T^{am}$ exchanges classical messages with each of $X^a$ and $Y^m$ 
				\item On input $\eta$, the actions of $(X^m, Y^m)$ and $(X^a\leftrightharpoons T^{am}, Y^m)$ are precisely the same, where $X^a\leftrightharpoons T^{am}$ denotes the joint system consisting of $X^a$ and $T^{am}$, regarded as one party to Blind PBC, i.e. they induce the same quantum channel. 
				\item Likewise, on input $\eta$, the actions of $(X^a, Y^a)$ and $(X^a, T^{am}\leftrightharpoons Y^m)$ also induce the same quantum channel, with notation as above, where $T^{am}\leftrightharpoons Y^m$ is regarded as one party to \ref{Prot:HiddenCompBasisAux}. 
			\end{enumerate}
		\end{enumerate}
		
		For the purpose of the definition above, we may assume that $X^a$ and $X^m$ have access to quantum register spaces of the same size (in particular, the interaction between $X^a$ and the classical machine $T^{am}$ has no effect on the size of the quantum register space of $X^a$) and likewise, $Y^a$ and $Y^m$ have access to respective quantum input and output register spaces of the same size. Informally, the purpose of $T^{am}$ is to interact with both $X^a$ and $Y^m$ (where $X^a$ and $Y^m$ otherwise do not interact with each other) such that:
		\begin{enumerate}[label=\roman*)]
			\itemsep-0.1em
			\item $T^{am}$ uses its interaction with $X^a$ to simulate for $Y^m$ an interaction with $X^m$ ; AND
			\item Simultaneously, $T^{am}$ uses its interaction with $Y^m$ to simulate for $X^a$ an interaction with $Y^a$.
		\end{enumerate}
		
		We claim now that if an AM translator exists, then \Cref{lemma:AuxToMainPerfSec} holds. Suppose there is such an AM translator $T^{am}$ and let $(\tilde{X}^m, \tilde{Y}^m)$ be an arbitrary initialization of Blind PBC, with the server $\tilde{Y}^m$ possibly deviating. Also let $\eta=\eta(C)$ be an initial shared input to Blind PBC as above. Let $\Phi_{\eta}:L(\mathcal{A}_{\tilde{X}^m}\otimes\mathcal{B}_{\tilde{Y}^m})\mapsto L(\mathcal{Z}^m)$ be the quantum channel induced by $(\tilde{X}^m, \tilde{Y}^m)$ on this input, where $\mathcal{Z}^m$ is the output space of $\tilde{Y}^m$. 
		
		Now suppose that on input $\eta$, we instead run the joint system $T^{am}\leftrightharpoons\tilde{Y}^m$ against some client $\tilde{X}^a$ that follows \ref{Prot:HiddenCompBasisAux}, i.e. $(\tilde{X}^a, T^{am}\leftrightharpoons\tilde{Y}^m)$ is an initialization for \ref{Prot:HiddenCompBasisAux}. We write $\tilde{Y}^a$ to denote the joint system $T^{am}\leftrightharpoons\tilde{Y}^m$ interacting with $\tilde{X}^a$ in \ref{Prot:HiddenCompBasisAux}. The interaction between $\tilde{X}^a$ and $\tilde{Y}^a$ on input $\eta$ induces a quantum channel $\Psi_{\eta}:L(\mathcal{A}_{\tilde{X}^a}\otimes\mathcal{B}_{\tilde{Y}^m})\mapsto L(\mathcal{Z}^m)$. In particular, $\tilde{Y}^m$ and $\tilde{Y}^a$ have the same quantum input registers (and the same output registers) and, by definition, we have
		\begin{equation}\label{DistEtaPsi}||\Phi_{\eta}-\Psi_{\eta}||_{\diamond}=0 \,.\end{equation}
		Our assumption on the security of \ref{Prot:HiddenCompBasisAux} implies that there exists a $\tilde{Y}^a$-simulator $\mathscr{S}_{\tilde{Y}^a}$ which, on input $\eta$, induces a channel $\Sigma_{\eta}:L(\mathcal{B}_{\tilde{Y}^m})\mapsto L(\mathcal{Z}^m)$ such that $||\Psi_{\eta}-\Sigma_{\eta}'||_{\diamond}=0$, where $\Sigma'_{\eta}:L(\mathcal{A}_{X^a}\otimes\mathcal{B}_{Y^m})\mapsto L(\mathcal{Z}^m)$ is as in \Cref{EpsilonSecurity}. Thus, by \cref{DistEtaPsi}, we have $||\Phi_{\eta}-\Sigma'_{\eta}||_{\diamond}=0$. It follows that $\mathscr{S}_{\tilde{Y}^a}$ also a $Y^m$-simulator. Thus, \ref{Prot:HiddenCompBasisAux} is indeed perfectly secure.
		
		To conclude the proof of \Cref{lemma:AuxToMainPerfSec}, we just need to construct an AM translator. We construct $T^{am}$ as follows: Suppose we have initial shared input $\eta=\eta(C)$ as above, and we take $(X^m, Y^m)$ and $(X^a, Y^a)$ as in \ref{defT2}, where client holds secret input $\mathbf{x}$. Now $T^{am}$ computes back-propagated Paulis specifying measurement bases and sequentially receives bits from the client from $(s_1, \cdots, s_t)$ as in \ref{enum:AuxProt3} of \ref{Prot:HiddenCompBasisAux}. Each time $T^{am}$ computes a back-propagated Pauli $P$, it sends a classical description of $P$ to $Y^m$, who performs the corresponding non-destructive measurement and sends a measurement outcome $\tilde{m}$ back to $T^{am}$ in return, which $T^{am}$ then sends to $X^a$, so that $X^a$ can use $\tilde{m}$ to compute the next bit in $(s_1, \cdots, s_t)$. They proceed in this way until $T^{am}$ has received all of $(s_1, \cdots, s_t)$, at which point there are, possibly, still more quantum measurements to be done by $Y^m$ (where these measurements are performed in Pauli bases that $Y^m$ can now determine independently). $Y^m$ performs the appropriate sequence of nondestructive measurements, sending the outcomes in sequence to $T^{am}$, who passes them in sequence along to $X^a$. \end{proof}
	
	Thus, to prove \Cref{MainSecurityThm}, it suffices to prove the security of \ref{Prot:HiddenCompBasisAux}. We first note that we can characterize the sequence of classical messages the server has received at each step in \ref{Prot:HiddenCompBasisAux}, where the Boolean function $f_i$ defined below is specified by the update rules in \Cref{alg:bpbc}.
	
	\begin{fact}\label{TBooleanFunctions} Let $(C,\mathbf{x})$ be an initial input to \ref{Prot:HiddenCompBasisAux}, with $n,t,w$ as specified and let $(\mathbf{c}, \mathbf{d})$ denote the client's initial $2t$-bit classical encryption key as applied to $|T\rangle^{\otimes t}$. Then, for each $1\leq i \leq t$, there exists a Boolean function $f_i:\{0,1\}^{n+3(i-1)}\mapsto\{0,1\}$, where $f_i$ depends only on $C,i$, such that $$s_i=d_i\oplus f_i(\mathbf{x}, c_1, \cdots, c_{i-1}, d_1, \cdots, d_{i-1}, m_1, \cdots, m_{i-1})\,,$$ where $m_1, \cdots, m_{i-1}$ are the gadget measurement outcomes received up to and including the end of step $i-1$.  \end{fact}
	
	To prove the security of \ref{Prot:HiddenCompBasisAux}, we first prove, using \Cref{TBooleanFunctions}, that, in a run of \ref{Prot:HiddenCompBasisAux}, at each step of the protocol, the server's view of the density matrix specified in \Cref{rhokfromcds} is maximally mixed. 
	
	\begin{lemma}\label{MaximallyMixedLemm} Let $(C,\mathbf{x})$ be an arbitrary input to \ref{Prot:HiddenCompBasisAux}, with $n,t,w$ as specified. For each $0\leq i<t$, let $\mathbf{s}_i\coloneqq(s_1, \cdots, s_i)$ be the string of bits received by the server up to the start of step $i$ of \ref{enum:AuxProt3} of \ref{Prot:HiddenCompBasisAux}. Then, for any $0\leq k<t$, from the server's point of view, $\rho_k(\mathbf{c}, \mathbf{d}|\mathbf{s}_k)$ is maximally mixed. More formally, given $\mathbf{s}_k$, we have
		$$\frac{1}{2^{2t}}\sum_{(\mathbf{c}, \mathbf{d})\in\mathbb{F}_2^{2t}}\rho_k\left(\mathbf{c}, \mathbf{d}\vert \mathbf{s}_k\right)=\frac{1}{2^{t+k}}I_{2^{t+k}}\,,$$
		where $\rho_k(\mathbf{c}, \mathbf{d}\vert\mathbf{s}_k)$ is as in \Cref{rhokfromcds}. \end{lemma}
	
	\begin{proof} We first introduce the following notation. For any $c,d\in\{0,1\}$, we let $\tau_{cd}$ be the density matrix corresponding to the pure state $S^dZ^c|T\rangle=\frac{1}{\sqrt{2}}\left(|0\rangle+(-1)^ce^{\frac{i\pi}{4}+\frac{i\pi d}{2}}|1\rangle\right)$. Now we prove \Cref{MaximallyMixedLemm} by induction on $k$. We first note that
		\begin{equation}
			\begin{split}
				\tau_{00}+\tau_{10}=I_2 \,,\\
				\tau_{01}+\tau_{11}=I_2\,.
			\end{split}
			\label{SumsToI2}
		\end{equation}
		So, for $k=0$, we indeed have
		\begin{equation*}
			\begin{split}
				\frac{1}{2^{2t}}\sum_{(\mathbf{c}, \mathbf{d})\in\mathbb{F}_2^{2t}}\rho_0(\tilde{c}, \tilde{d}) &=\frac{1}{2^{2t}}\left(\tau_{00}+\tau_{01}+\tau_{10}+\tau_{11}\right)^{\otimes t}\\
				&=\frac{1}{2^{2t}}\left(2I_2\right)^{\otimes t}=\frac{1}{2^t}I_{2^t} \,,
			\end{split}
		\end{equation*}
		as desired. Now let $0<k\leq t$ and suppose the induction hypothesis holds for any $0\leq i<k$. Suppose we are at the start of step $k$ of \ref{enum:AuxProt3} of \ref{Prot:HiddenCompBasisAux}. Fix an arbitrary $(\mathbf{c}, \mathbf{d})\in\mathbb{F}_2^{2t}$, i.e., the client's initial classical encryption key for the rotation of $|T\rangle^{\otimes t}$. Let $(C,\mathbf{x})$ and $f_{k}:\{0,1\}^{n+3k}\mapsto\{0,1\}$ be as in \Cref{TBooleanFunctions}. We then have
		\begin{equation*}
			\rho_{k}(\mathbf{c}, \mathbf{d}\vert \mathbf{s}_{k})=\rho_{k-1}(\mathbf{c}, \mathbf{d}|\mathbf{s}_{k-1})\otimes |f_k\oplus d_k\rangle\langle f_k\oplus d_k|
		\end{equation*}
		where we have suppressed the arguments to $f_{k}$. We note that, following \Cref{TBooleanFunctions}, the terms $c_k, d_k$ only appear in the subsystems of index $k$ and $t+k$ of $\rho_{k}(\mathbf{c}, \mathbf{d}\vert \mathbf{s}_{k})$. Furthermore, the bit $f_k$ only depends on $\mathbf{x}, c_1, \cdots, c_{k-1}, d_1, \cdots, d_{k-1}, m_1, \cdots, m_{k-1}$, so it does not depend on $d_k, c_k$. We write $\sigma=\sigma\left(c_k, d_k|\mathbf{s}_k\right)$ to denote the 2-qubit density matrix obtained from the above expression by tracing out all but subsystems $k, t+k$, with the notation emphasizing that this subsystem only depends on $c_k, d_k, \mathbf{s}_k$. We then have the following:
		\begin{itemize}
			\itemsep-0.2em
			\item If $c_k, d_k=0$, then $\sigma=\tau_{00}\otimes|f_k\rangle\langle f_k|\,.$
			\item If $c_k=0$ and $d_k=1$, then $\sigma=\tau_{01}\otimes|\overline{f_k}\rangle\langle\overline{f_k}|\,.$
			\item If $c_k=1$ and $d_k=0$, then $\sigma=\tau_{10}\otimes |f_k\rangle\langle f_k|\,.$
			\item If $c_k=d_k=1$, then $\sigma=\tau_{11}\otimes |\overline{f_k}\rangle\langle\overline{f_k}|\,.$
		\end{itemize}
		Thus, by \Cref{SumsToI2}, we get
		\begin{equation}
			\begin{split}
				&\sum_{c,d\in\{0,1\}}\sigma\left(c,d|\mathbf{s}_k\right)\\
				&=\tau_{00}\otimes |f_k\rangle\langle f_k|+\tau_{01}\otimes |\overline{f_k}\rangle\langle \overline{f_k}| \\ 
				&+\tau_{10}\otimes |f_k\rangle\langle f_k|+\tau_{11}\otimes  |\overline{f_k}\rangle\langle \overline{f_k}| \\
				&=(\tau_{00}+\tau_{10})\otimes |f_k\rangle\langle f_k|+(\tau_{01}+\tau_{11})\otimes |\overline{f_k}\rangle\langle \overline{f_k}| \\ &=I_2\otimes|f_k\rangle\langle f_k|+I_2\otimes |\overline{f_k}\rangle\langle \overline{f_k}| \\
				&=I_4 \,.
			\end{split}
			\label{splitlineidentity2}
		\end{equation}
		The sequence of equalities in \Cref{splitlineidentity2} holds regardless of the probability distribution on $\{0,1\}$ induced by $f_k$. Now, we turn our attention to $\rho_{k-1}(\mathbf{c}, \mathbf{d}\vert\mathbf{s}_{k-1})$. If we trace out the $k$th qubit, the resulting subsystem no longer depends on $c_k, d_k$ (again by \Cref{TBooleanFunctions}) so it is indexed by a vector $\left(\mathbf{a}, \mathbf{b}\right)$ where $\mathbf{a}, \mathbf{b}\in\mathbb{F}_2^{t-1}$. We write $\rho'\left(\mathbf{a}, \mathbf{b}\vert\mathbf{s}_{k-1}\right)$ to denote this $(t+k-2)$-qubit system. Then, by induction, we have
		\begin{equation}\sum_{(\mathbf{a}, \mathbf{b})\in\mathbb{F}_2^{2(t-1)}}\rho'_{k-1}\left(\mathbf{a}, \mathbf{b}|\mathbf{s}_{k-1}\right)=\frac{2^{2(t-1)}}{2^{t+k-2}}I_{2^{t+k-2}}\,.
			\label{labelsplitlineover2t-1}
		\end{equation}
		
		Summing $\rho_k(\mathbf{c}, \mathbf{d}|\mathbf{s}_k)$ over all the terms of $\mathbb{F}_2^{2t}$, and using \Cref{splitlineidentity2} and \Cref{labelsplitlineover2t-1}, we obtain the following, where the second line holds up to permutation of the tensor product factors.
		\begin{equation*}
			\begin{split}
				&\sum_{(\mathbf{c}, \mathbf{d})\in\mathbb{F}_2^{2t}}\rho_k\left(\mathbf{c}, \mathbf{d}\vert \mathbf{s}_k\right)= \\ 
				&\left(\sum_{\left(\mathbf{a}, \mathbf{b}\right)\in\mathbb{F}_2^{2(t-1)}}\rho'_{k-1}\left(\mathbf{a}, \mathbf{b}\vert\mathbf{s}_{k-1}\right)\right)\otimes\left(\sum_{c,d\in\{0,1\}}\sigma\left(c,d\vert\mathbf{s}_k\right)\right) \\
				&=\left(\frac{2^{2(t-1)}}{2^{t+k-2}}I_{2^{t+k-2}}\right)\otimes I_4=\frac{2^{2t}}{2^{t+k}}I_{2^{t+k}}\,,
			\end{split}
		\end{equation*}
		as desired. \end{proof}
	
	With the above lemmas in hand, the proof of security of \ref{Prot:HiddenCompBasisAux} is relatively straightforward.
	
	\begin{theorem}\label{Pro2AProof} \ref{Prot:HiddenCompBasisAux} is perfectly secure. \end{theorem}
	
	\begin{proof} Let $(X, Y)$ be an arbitrary initialization for \ref{Prot:HiddenCompBasisAux} and let $\delta\in\{0,1\}^*$ be a bit string encoding a shared protocol input circuit. Let $\Phi_{\delta}:L(\mathcal{A}_X\otimes\mathcal{B}_Y)\mapsto L(\mathcal{Z})$ be the corresponding channel induced by this potentially deviating server. We define $\mathscr{S}_{Y}$ to be the algorithm which, on input $\delta\in\{0,1\}^*$, induces the channel $\Sigma_{\delta}:L(
		\mathcal{B}_Y)\mapsto L(\mathcal{Z})$ which, on input a $\tau\in L(\mathcal{B}_Y)$ consists of applying $\Phi_{\delta}$ to $\mu\otimes\tau$, where $\mu\in L(\mathcal{A}_X)$ is maximally mixed. By Lemma \ref{MaximallyMixedLemm}, we have $||\Sigma'_{\delta}-\Phi_{\delta}||_{\diamond}=0$, where $\Sigma'_{\delta}$ is as in Definition \ref{EpsilonSecurity}, so we are done. \end{proof}
	
	The proof of \Cref{Pro2AProof} requires some care with the formalism, since the size of the state specified in \Cref{TBooleanFunctions} depends on the current step of the interactive protocol between $X$ and $Y$: We take $L(\mathcal{A}_X)$ to be sufficiently large that it contains enough registers for both the initial density matrix and all of the (classical) messages exchanged, regarded as computational basis states. By padding each state in \Cref{rhokfromcds} with extra registers which are initially maximally mixed, we can regard the interactive part of $\Phi_{\delta}$ as executing a sequence of steps where each interactive step between the client and server only involves operations that have no effect on the dimension of the current shared density matrix, as all measurements are nondestructive. Then, \Cref{TBooleanFunctions} states that, if, at each such step, we discard all the $\mathcal{A}_X$-registers of the current shared density matrix of $L(\mathcal{A}_X\otimes\mathcal{B}_Y)$ and replace them with the maximally mixed state on the $\mathcal{A}_X$-registers, this has no effect on the final output in the server's output space. We note that, although we have followed the security definition of~\cite{Bro15}, our security proof is simpler because our protocol, unlike that of~\cite{Bro15}, does not require us to analyze an \emph{entanglement-based} auxiliary protocol in which the density matrix the client sends to the server is correlated with some additional registers kept only by her. 
	

\begin{thebibliography}{31}%
		\makeatletter
		\providecommand \@ifxundefined [1]{%
			\@ifx{#1\undefined}
		}%
		\providecommand \@ifnum [1]{%
			\ifnum #1\expandafter \@firstoftwo
			\else \expandafter \@secondoftwo
			\fi
		}%
		\providecommand \@ifx [1]{%
			\ifx #1\expandafter \@firstoftwo
			\else \expandafter \@secondoftwo
			\fi
		}%
		\providecommand \natexlab [1]{#1}%
		\providecommand \enquote  [1]{``#1''}%
		\providecommand \bibnamefont  [1]{#1}%
		\providecommand \bibfnamefont [1]{#1}%
		\providecommand \citenamefont [1]{#1}%
		\providecommand \href@noop [0]{\@secondoftwo}%
		\providecommand \href [0]{\begingroup \@sanitize@url \@href}%
		\providecommand \@href[1]{\@@startlink{#1}\@@href}%
		\providecommand \@@href[1]{\endgroup#1\@@endlink}%
		\providecommand \@sanitize@url [0]{\catcode `\\12\catcode `\$12\catcode `\&12\catcode `\#12\catcode `\^12\catcode `\_12\catcode `\%12\relax}%
		\providecommand \@@startlink[1]{}%
		\providecommand \@@endlink[0]{}%
		\providecommand \url  [0]{\begingroup\@sanitize@url \@url }%
		\providecommand \@url [1]{\endgroup\@href {#1}{\urlprefix }}%
		\providecommand \urlprefix  [0]{URL }%
		\providecommand \Eprint [0]{\href }%
		\providecommand \doibase [0]{https://doi.org/}%
		\providecommand \selectlanguage [0]{\@gobble}%
		\providecommand \bibinfo  [0]{\@secondoftwo}%
		\providecommand \bibfield  [0]{\@secondoftwo}%
		\providecommand \translation [1]{[#1]}%
		\providecommand \BibitemOpen [0]{}%
		\providecommand \bibitemStop [0]{}%
		\providecommand \bibitemNoStop [0]{.\EOS\space}%
		\providecommand \EOS [0]{\spacefactor3000\relax}%
		\providecommand \BibitemShut  [1]{\csname bibitem#1\endcsname}%
		\let\auto@bib@innerbib\@empty
		\bibitem [{\citenamefont {Preskill}(2018)}]{Pre18}%
		\BibitemOpen
		\bibfield  {author} {\bibinfo {author} {\bibfnamefont {J.}~\bibnamefont {Preskill}},\ }\bibfield  {title} {\bibinfo {title} {Quantum {C}omputing in the {NISQ} era and beyond},\ }\href {https://doi.org/10.22331/q-2018-08-06-79} {\bibfield  {journal} {\bibinfo  {journal} {Quant.}\ }\textbf {\bibinfo {volume} {2}},\ \bibinfo {pages} {79} (\bibinfo {year} {2018})}\BibitemShut {NoStop}%
		\bibitem [{\citenamefont {Childs}(2005)}]{Chi05}%
		\BibitemOpen
		\bibfield  {author} {\bibinfo {author} {\bibfnamefont {A.~M.}\ \bibnamefont {Childs}},\ }\bibfield  {title} {\bibinfo {title} {Secure assisted quantum computation},\ }\href {https://doi.org/10.26421/QIC5.6-4} {\bibfield  {journal} {\bibinfo  {journal} {Quant. Inf. Comp.}\ }\textbf {\bibinfo {volume} {5}},\ \bibinfo {pages} {456} (\bibinfo {year} {2005})}\BibitemShut {NoStop}%
		\bibitem [{\citenamefont {Broadbent}\ \emph {et~al.}(2009)\citenamefont {Broadbent}, \citenamefont {Fitzsimons},\ and\ \citenamefont {Kashefi}}]{BFK09}%
		\BibitemOpen
		\bibfield  {author} {\bibinfo {author} {\bibfnamefont {A.}~\bibnamefont {Broadbent}}, \bibinfo {author} {\bibfnamefont {J.}~\bibnamefont {Fitzsimons}},\ and\ \bibinfo {author} {\bibfnamefont {E.}~\bibnamefont {Kashefi}},\ }\bibfield  {title} {\bibinfo {title} {Universal {B}lind {Q}uantum {C}omputation},\ }in\ \href {https://doi.org/10.1109/FOCS.2009.36} {\emph {\bibinfo {booktitle} {FOCS 2009}}}\ (\bibinfo {year} {2009})\ pp.\ \bibinfo {pages} {517--526}\BibitemShut {NoStop}%
		\bibitem [{\citenamefont {Broadbent}(2015)}]{Bro15}%
		\BibitemOpen
		\bibfield  {author} {\bibinfo {author} {\bibfnamefont {A.}~\bibnamefont {Broadbent}},\ }\bibfield  {title} {\bibinfo {title} {Delegating private quantum computations},\ }\href {https://doi.org/10.1139/cjp-2015-0030} {\bibfield  {journal} {\bibinfo  {journal} {Can. J. Phys.}\ }\textbf {\bibinfo {volume} {93}},\ \bibinfo {pages} {941} (\bibinfo {year} {2015})}\BibitemShut {NoStop}%
		\bibitem [{\citenamefont {Aaronson}\ \emph {et~al.}(2019)\citenamefont {Aaronson}, \citenamefont {Cojocaru}, \citenamefont {Gheorghiu},\ and\ \citenamefont {Kashefi}}]{ACGK19}%
		\BibitemOpen
		\bibfield  {author} {\bibinfo {author} {\bibfnamefont {S.}~\bibnamefont {Aaronson}}, \bibinfo {author} {\bibfnamefont {A.}~\bibnamefont {Cojocaru}}, \bibinfo {author} {\bibfnamefont {A.}~\bibnamefont {Gheorghiu}},\ and\ \bibinfo {author} {\bibfnamefont {E.}~\bibnamefont {Kashefi}},\ }\bibfield  {title} {\bibinfo {title} {Complexity-{T}heoretic {L}imitations on {B}lind {D}elegated {Q}uantum {C}omputation},\ }in\ \href {https://doi.org/10.4230/LIPIcs.ICALP.2019.6} {\emph {\bibinfo {booktitle} {ICALP2019}}}\ (\bibinfo {year} {2019})\ pp.\ \bibinfo {pages} {6:1--6:13}\BibitemShut {NoStop}%
		\bibitem [{\citenamefont {Mantri}\ \emph {et~al.}(2013)\citenamefont {Mantri}, \citenamefont {P{\'e}rez-Delgado},\ and\ \citenamefont {Fitzsimons}}]{MPDF13}%
		\BibitemOpen
		\bibfield  {author} {\bibinfo {author} {\bibfnamefont {A.}~\bibnamefont {Mantri}}, \bibinfo {author} {\bibfnamefont {C.~A.}\ \bibnamefont {P{\'e}rez-Delgado}},\ and\ \bibinfo {author} {\bibfnamefont {J.~F.}\ \bibnamefont {Fitzsimons}},\ }\bibfield  {title} {\bibinfo {title} {Optimal {B}lind {Q}uantum {C}omputation},\ }\href {https://doi.org/10.1103/PhysRevLett.111.230502} {\bibfield  {journal} {\bibinfo  {journal} {Phys. Rev. Lett.}\ }\textbf {\bibinfo {volume} {111}},\ \bibinfo {pages} {230502} (\bibinfo {year} {2013})}\BibitemShut {NoStop}%
		\bibitem [{\citenamefont {Davies}\ and\ \citenamefont {Kay}(2025)}]{DK25arxiv}%
		\BibitemOpen
		\bibfield  {author} {\bibinfo {author} {\bibfnamefont {E.}~\bibnamefont {Davies}}\ and\ \bibinfo {author} {\bibfnamefont {A.}~\bibnamefont {Kay}},\ }\href@noop {} {\bibinfo {title} {Communication-{O}ptimal {B}lind {Q}uantum {P}rotocols}},\ \bibinfo {howpublished} {E-print arXiv:2510.07112 [quant-ph]} (\bibinfo {year} {2025})\BibitemShut {NoStop}%
		\bibitem [{\citenamefont {Drmota}\ \emph {et~al.}(2024)\citenamefont {Drmota}, \citenamefont {Nadlinger}, \citenamefont {Main}, \citenamefont {Nichol}, \citenamefont {Ainley}, \citenamefont {Leichtle}, \citenamefont {Mantri}, \citenamefont {Kashefi}, \citenamefont {Srinivas}, \citenamefont {Araneda}, \citenamefont {Ballance},\ and\ \citenamefont {Lucas}}]{DNM+24}%
		\BibitemOpen
		\bibfield  {author} {\bibinfo {author} {\bibfnamefont {P.}~\bibnamefont {Drmota}}, \bibinfo {author} {\bibfnamefont {D.~P.}\ \bibnamefont {Nadlinger}}, \bibinfo {author} {\bibfnamefont {D.}~\bibnamefont {Main}}, \bibinfo {author} {\bibfnamefont {B.~C.}\ \bibnamefont {Nichol}}, \bibinfo {author} {\bibfnamefont {E.~M.}\ \bibnamefont {Ainley}}, \bibinfo {author} {\bibfnamefont {D.}~\bibnamefont {Leichtle}}, \bibinfo {author} {\bibfnamefont {A.}~\bibnamefont {Mantri}}, \bibinfo {author} {\bibfnamefont {E.}~\bibnamefont {Kashefi}}, \bibinfo {author} {\bibfnamefont {R.}~\bibnamefont {Srinivas}}, \bibinfo {author} {\bibfnamefont {G.}~\bibnamefont {Araneda}}, \bibinfo {author} {\bibfnamefont {C.~J.}\ \bibnamefont {Ballance}},\ and\ \bibinfo {author} {\bibfnamefont {D.~M.}\ \bibnamefont {Lucas}},\ }\bibfield  {title} {\bibinfo {title} {Verifiable {B}lind {Q}uantum {C}omputing with {T}rapped {I}ons and {S}ingle {P}hotons},\ }\href {https://doi.org/10.1103/PhysRevLett.132.150604} {\bibfield  {journal} {\bibinfo
				{journal} {Phys. Rev. Lett.}\ }\textbf {\bibinfo {volume} {132}},\ \bibinfo {pages} {150604} (\bibinfo {year} {2024})}\BibitemShut {NoStop}%
		\bibitem [{\citenamefont {Wei}\ \emph {et~al.}(2025)\citenamefont {Wei}, \citenamefont {Stas}, \citenamefont {Suleymanzade}, \citenamefont {Baranes}, \citenamefont {Machado}, \citenamefont {Huan}, \citenamefont {Knaut}, \citenamefont {Ding}, \citenamefont {Merz}, \citenamefont {Knall}, \citenamefont {Yazlar}, \citenamefont {Sirotin}, \citenamefont {Wang}, \citenamefont {Machielse}, \citenamefont {Yelin}, \citenamefont {Borregaard}, \citenamefont {Park}, \citenamefont {Lon{\v{c}}ar},\ and\ \citenamefont {Lukin}}]{WSS+25}%
		\BibitemOpen
		\bibfield  {author} {\bibinfo {author} {\bibfnamefont {Y.-C.}\ \bibnamefont {Wei}}, \bibinfo {author} {\bibfnamefont {P.-J.}\ \bibnamefont {Stas}}, \bibinfo {author} {\bibfnamefont {A.}~\bibnamefont {Suleymanzade}}, \bibinfo {author} {\bibfnamefont {G.}~\bibnamefont {Baranes}}, \bibinfo {author} {\bibfnamefont {F.}~\bibnamefont {Machado}}, \bibinfo {author} {\bibfnamefont {Y.~Q.}\ \bibnamefont {Huan}}, \bibinfo {author} {\bibfnamefont {C.~M.}\ \bibnamefont {Knaut}}, \bibinfo {author} {\bibfnamefont {S.~W.}\ \bibnamefont {Ding}}, \bibinfo {author} {\bibfnamefont {M.}~\bibnamefont {Merz}}, \bibinfo {author} {\bibfnamefont {E.~N.}\ \bibnamefont {Knall}}, \bibinfo {author} {\bibfnamefont {U.}~\bibnamefont {Yazlar}}, \bibinfo {author} {\bibfnamefont {M.}~\bibnamefont {Sirotin}}, \bibinfo {author} {\bibfnamefont {I.~W.}\ \bibnamefont {Wang}}, \bibinfo {author} {\bibfnamefont {B.}~\bibnamefont {Machielse}}, \bibinfo {author} {\bibfnamefont {S.~F.}\ \bibnamefont {Yelin}}, \bibinfo {author} {\bibfnamefont
				{J.}~\bibnamefont {Borregaard}}, \bibinfo {author} {\bibfnamefont {H.}~\bibnamefont {Park}}, \bibinfo {author} {\bibfnamefont {M.}~\bibnamefont {Lon{\v{c}}ar}},\ and\ \bibinfo {author} {\bibfnamefont {M.~D.}\ \bibnamefont {Lukin}},\ }\bibfield  {title} {\bibinfo {title} {Universal distributed blind quantum computing with solid-state qubits},\ }\href {https://doi.org/10.1126/science.adu6894} {\bibfield  {journal} {\bibinfo  {journal} {Sci.}\ }\textbf {\bibinfo {volume} {388}},\ \bibinfo {pages} {509} (\bibinfo {year} {2025})}\BibitemShut {NoStop}%
		\bibitem [{\citenamefont {Debroux}\ \emph {et~al.}(2021)\citenamefont {Debroux}, \citenamefont {Michaels}, \citenamefont {Purser}, \citenamefont {Wan}, \citenamefont {Trusheim}, \citenamefont {Arjona~Mart\'{\i}nez}, \citenamefont {Parker}, \citenamefont {Stramma}, \citenamefont {Chen}, \citenamefont {de~Santis}, \citenamefont {Alexeev}, \citenamefont {Ferrari}, \citenamefont {Englund}, \citenamefont {Gangloff},\ and\ \citenamefont {Atat\"ure}}]{DMP+21}%
		\BibitemOpen
		\bibfield  {author} {\bibinfo {author} {\bibfnamefont {R.}~\bibnamefont {Debroux}}, \bibinfo {author} {\bibfnamefont {C.~P.}\ \bibnamefont {Michaels}}, \bibinfo {author} {\bibfnamefont {C.~M.}\ \bibnamefont {Purser}}, \bibinfo {author} {\bibfnamefont {N.}~\bibnamefont {Wan}}, \bibinfo {author} {\bibfnamefont {M.~E.}\ \bibnamefont {Trusheim}}, \bibinfo {author} {\bibfnamefont {J.}~\bibnamefont {Arjona~Mart\'{\i}nez}}, \bibinfo {author} {\bibfnamefont {R.~A.}\ \bibnamefont {Parker}}, \bibinfo {author} {\bibfnamefont {A.~M.}\ \bibnamefont {Stramma}}, \bibinfo {author} {\bibfnamefont {K.~C.}\ \bibnamefont {Chen}}, \bibinfo {author} {\bibfnamefont {L.}~\bibnamefont {de~Santis}}, \bibinfo {author} {\bibfnamefont {E.~M.}\ \bibnamefont {Alexeev}}, \bibinfo {author} {\bibfnamefont {A.~C.}\ \bibnamefont {Ferrari}}, \bibinfo {author} {\bibfnamefont {D.}~\bibnamefont {Englund}}, \bibinfo {author} {\bibfnamefont {D.~A.}\ \bibnamefont {Gangloff}},\ and\ \bibinfo {author} {\bibfnamefont {M.}~\bibnamefont {Atat\"ure}},\
		}\bibfield  {title} {\bibinfo {title} {Quantum {C}ontrol of the {T}in-{V}acancy {S}pin {Q}ubit in {D}iamond},\ }\href {https://doi.org/10.1103/PhysRevX.11.041041} {\bibfield  {journal} {\bibinfo  {journal} {Phys. Rev. X}\ }\textbf {\bibinfo {volume} {11}},\ \bibinfo {pages} {041041} (\bibinfo {year} {2021})}\BibitemShut {NoStop}%
		\bibitem [{\citenamefont {Abobeih}\ \emph {et~al.}(2022)\citenamefont {Abobeih}, \citenamefont {Wang}, \citenamefont {Randall}, \citenamefont {Loenen}, \citenamefont {Bradley}, \citenamefont {Markham}, \citenamefont {Twitchen}, \citenamefont {Terhal},\ and\ \citenamefont {Taminiau}}]{AWR+22}%
		\BibitemOpen
		\bibfield  {author} {\bibinfo {author} {\bibfnamefont {M.~H.}\ \bibnamefont {Abobeih}}, \bibinfo {author} {\bibfnamefont {Y.}~\bibnamefont {Wang}}, \bibinfo {author} {\bibfnamefont {J.}~\bibnamefont {Randall}}, \bibinfo {author} {\bibfnamefont {S.~J.~H.}\ \bibnamefont {Loenen}}, \bibinfo {author} {\bibfnamefont {C.~E.}\ \bibnamefont {Bradley}}, \bibinfo {author} {\bibfnamefont {M.}~\bibnamefont {Markham}}, \bibinfo {author} {\bibfnamefont {D.~J.}\ \bibnamefont {Twitchen}}, \bibinfo {author} {\bibfnamefont {B.~M.}\ \bibnamefont {Terhal}},\ and\ \bibinfo {author} {\bibfnamefont {T.~H.}\ \bibnamefont {Taminiau}},\ }\bibfield  {title} {\bibinfo {title} {Fault-tolerant operation of a logical qubit in a diamond quantum processor},\ }\href {https://doi.org/10.1038/s41586-022-04819-6} {\bibfield  {journal} {\bibinfo  {journal} {Nat.}\ }\textbf {\bibinfo {volume} {606}},\ \bibinfo {pages} {884} (\bibinfo {year} {2022})}\BibitemShut {NoStop}%
		\bibitem [{\citenamefont {Pino}\ \emph {et~al.}(2021)\citenamefont {Pino}, \citenamefont {Dreiling}, \citenamefont {Figgatt}, \citenamefont {Gaebler}, \citenamefont {Moses}, \citenamefont {Allman}, \citenamefont {Baldwin}, \citenamefont {Foss-Feig}, \citenamefont {Hayes}, \citenamefont {Mayer}, \citenamefont {Ryan-Anderson},\ and\ \citenamefont {Neyenhuis}}]{PDF+21}%
		\BibitemOpen
		\bibfield  {author} {\bibinfo {author} {\bibfnamefont {J.~M.}\ \bibnamefont {Pino}}, \bibinfo {author} {\bibfnamefont {J.~M.}\ \bibnamefont {Dreiling}}, \bibinfo {author} {\bibfnamefont {C.}~\bibnamefont {Figgatt}}, \bibinfo {author} {\bibfnamefont {J.~P.}\ \bibnamefont {Gaebler}}, \bibinfo {author} {\bibfnamefont {S.~A.}\ \bibnamefont {Moses}}, \bibinfo {author} {\bibfnamefont {M.~S.}\ \bibnamefont {Allman}}, \bibinfo {author} {\bibfnamefont {C.~H.}\ \bibnamefont {Baldwin}}, \bibinfo {author} {\bibfnamefont {M.}~\bibnamefont {Foss-Feig}}, \bibinfo {author} {\bibfnamefont {D.}~\bibnamefont {Hayes}}, \bibinfo {author} {\bibfnamefont {K.}~\bibnamefont {Mayer}}, \bibinfo {author} {\bibfnamefont {C.}~\bibnamefont {Ryan-Anderson}},\ and\ \bibinfo {author} {\bibfnamefont {B.}~\bibnamefont {Neyenhuis}},\ }\bibfield  {title} {\bibinfo {title} {Demonstration of the trapped-ion quantum {CCD} computer architecture},\ }\href {https://doi.org/10.1038/s41586-021-03318-4} {\bibfield  {journal} {\bibinfo  {journal} {Nat.}\
			}\textbf {\bibinfo {volume} {592}},\ \bibinfo {pages} {209} (\bibinfo {year} {2021})}\BibitemShut {NoStop}%
		\bibitem [{\citenamefont {Moses}\ \emph {et~al.}(2023)\citenamefont {Moses}, \citenamefont {Baldwin}, \citenamefont {Allman}, \citenamefont {Ancona}, \citenamefont {Ascarrunz}, \citenamefont {Barnes}, \citenamefont {Bartolotta}, \citenamefont {Bjork}, \citenamefont {Blanchard}, \citenamefont {Bohn}, \citenamefont {Bohnet}, \citenamefont {Brown}, \citenamefont {Burdick}, \citenamefont {Burton}, \citenamefont {Campbell}, \citenamefont {Campora}, \citenamefont {Carron}, \citenamefont {Chambers}, \citenamefont {Chan}, \citenamefont {Chen}, \citenamefont {Chernoguzov}, \citenamefont {Chertkov}, \citenamefont {Colina}, \citenamefont {Curtis}, \citenamefont {Daniel}, \citenamefont {DeCross}, \citenamefont {Deen}, \citenamefont {Delaney}, \citenamefont {Dreiling}, \citenamefont {Ertsgaard}, \citenamefont {Esposito}, \citenamefont {Estey}, \citenamefont {Fabrikant}, \citenamefont {Figgatt}, \citenamefont {Foltz}, \citenamefont {Foss-Feig}, \citenamefont {Francois}, \citenamefont {Gaebler}, \citenamefont {Gatterman},
			\citenamefont {Gilbreth}, \citenamefont {Giles}, \citenamefont {Glynn}, \citenamefont {Hall}, \citenamefont {Hankin}, \citenamefont {Hansen}, \citenamefont {Hayes}, \citenamefont {Higashi}, \citenamefont {Hoffman}, \citenamefont {Horning}, \citenamefont {Hout}, \citenamefont {Jacobs}, \citenamefont {Johansen}, \citenamefont {Jones}, \citenamefont {Karcz}, \citenamefont {Klein}, \citenamefont {Lauria}, \citenamefont {Lee}, \citenamefont {Liefer}, \citenamefont {Lu}, \citenamefont {Lucchetti}, \citenamefont {Lytle}, \citenamefont {Malm}, \citenamefont {Matheny}, \citenamefont {Mathewson}, \citenamefont {Mayer}, \citenamefont {Miller}, \citenamefont {Mills}, \citenamefont {Neyenhuis}, \citenamefont {Nugent}, \citenamefont {Olson}, \citenamefont {Parks}, \citenamefont {Price}, \citenamefont {Price}, \citenamefont {Pugh}, \citenamefont {Ransford}, \citenamefont {Reed}, \citenamefont {Roman}, \citenamefont {Rowe}, \citenamefont {Ryan-Anderson}, \citenamefont {Sanders}, \citenamefont {Sedlacek}, \citenamefont
			{Shevchuk}, \citenamefont {Siegfried}, \citenamefont {Skripka}, \citenamefont {Spaun}, \citenamefont {Sprenkle}, \citenamefont {Stutz}, \citenamefont {Swallows}, \citenamefont {Tobey}, \citenamefont {Tran}, \citenamefont {Tran}, \citenamefont {Vogt}, \citenamefont {Volin}, \citenamefont {Walker}, \citenamefont {Zolot},\ and\ \citenamefont {Pino}}]{MBA+23}%
		\BibitemOpen
		\bibfield  {author} {\bibinfo {author} {\bibfnamefont {S.~A.}\ \bibnamefont {Moses}}, \bibinfo {author} {\bibfnamefont {C.~H.}\ \bibnamefont {Baldwin}}, \bibinfo {author} {\bibfnamefont {M.~S.}\ \bibnamefont {Allman}}, \bibinfo {author} {\bibfnamefont {R.}~\bibnamefont {Ancona}}, \bibinfo {author} {\bibfnamefont {L.}~\bibnamefont {Ascarrunz}}, \bibinfo {author} {\bibfnamefont {C.}~\bibnamefont {Barnes}}, \bibinfo {author} {\bibfnamefont {J.}~\bibnamefont {Bartolotta}}, \bibinfo {author} {\bibfnamefont {B.}~\bibnamefont {Bjork}}, \bibinfo {author} {\bibfnamefont {P.}~\bibnamefont {Blanchard}}, \bibinfo {author} {\bibfnamefont {M.}~\bibnamefont {Bohn}}, \bibinfo {author} {\bibfnamefont {J.~G.}\ \bibnamefont {Bohnet}}, \bibinfo {author} {\bibfnamefont {N.~C.}\ \bibnamefont {Brown}}, \bibinfo {author} {\bibfnamefont {N.~Q.}\ \bibnamefont {Burdick}}, \bibinfo {author} {\bibfnamefont {W.~C.}\ \bibnamefont {Burton}}, \bibinfo {author} {\bibfnamefont {S.~L.}\ \bibnamefont {Campbell}}, \bibinfo {author}
			{\bibfnamefont {J.~P.}\ \bibnamefont {Campora}}, \bibinfo {author} {\bibfnamefont {C.}~\bibnamefont {Carron}}, \bibinfo {author} {\bibfnamefont {J.}~\bibnamefont {Chambers}}, \bibinfo {author} {\bibfnamefont {J.~W.}\ \bibnamefont {Chan}}, \bibinfo {author} {\bibfnamefont {Y.~H.}\ \bibnamefont {Chen}}, \bibinfo {author} {\bibfnamefont {A.}~\bibnamefont {Chernoguzov}}, \bibinfo {author} {\bibfnamefont {E.}~\bibnamefont {Chertkov}}, \bibinfo {author} {\bibfnamefont {J.}~\bibnamefont {Colina}}, \bibinfo {author} {\bibfnamefont {J.~P.}\ \bibnamefont {Curtis}}, \bibinfo {author} {\bibfnamefont {R.}~\bibnamefont {Daniel}}, \bibinfo {author} {\bibfnamefont {M.}~\bibnamefont {DeCross}}, \bibinfo {author} {\bibfnamefont {D.}~\bibnamefont {Deen}}, \bibinfo {author} {\bibfnamefont {C.}~\bibnamefont {Delaney}}, \bibinfo {author} {\bibfnamefont {J.~M.}\ \bibnamefont {Dreiling}}, \bibinfo {author} {\bibfnamefont {C.~T.}\ \bibnamefont {Ertsgaard}}, \bibinfo {author} {\bibfnamefont {J.}~\bibnamefont {Esposito}}, \bibinfo
			{author} {\bibfnamefont {B.}~\bibnamefont {Estey}}, \bibinfo {author} {\bibfnamefont {M.}~\bibnamefont {Fabrikant}}, \bibinfo {author} {\bibfnamefont {C.}~\bibnamefont {Figgatt}}, \bibinfo {author} {\bibfnamefont {C.}~\bibnamefont {Foltz}}, \bibinfo {author} {\bibfnamefont {M.}~\bibnamefont {Foss-Feig}}, \bibinfo {author} {\bibfnamefont {D.}~\bibnamefont {Francois}}, \bibinfo {author} {\bibfnamefont {J.~P.}\ \bibnamefont {Gaebler}}, \bibinfo {author} {\bibfnamefont {T.~M.}\ \bibnamefont {Gatterman}}, \bibinfo {author} {\bibfnamefont {C.~N.}\ \bibnamefont {Gilbreth}}, \bibinfo {author} {\bibfnamefont {J.}~\bibnamefont {Giles}}, \bibinfo {author} {\bibfnamefont {E.}~\bibnamefont {Glynn}}, \bibinfo {author} {\bibfnamefont {A.}~\bibnamefont {Hall}}, \bibinfo {author} {\bibfnamefont {A.~M.}\ \bibnamefont {Hankin}}, \bibinfo {author} {\bibfnamefont {A.}~\bibnamefont {Hansen}}, \bibinfo {author} {\bibfnamefont {D.}~\bibnamefont {Hayes}}, \bibinfo {author} {\bibfnamefont {B.}~\bibnamefont {Higashi}}, \bibinfo
			{author} {\bibfnamefont {I.~M.}\ \bibnamefont {Hoffman}}, \bibinfo {author} {\bibfnamefont {B.}~\bibnamefont {Horning}}, \bibinfo {author} {\bibfnamefont {J.~J.}\ \bibnamefont {Hout}}, \bibinfo {author} {\bibfnamefont {R.}~\bibnamefont {Jacobs}}, \bibinfo {author} {\bibfnamefont {J.}~\bibnamefont {Johansen}}, \bibinfo {author} {\bibfnamefont {L.}~\bibnamefont {Jones}}, \bibinfo {author} {\bibfnamefont {J.}~\bibnamefont {Karcz}}, \bibinfo {author} {\bibfnamefont {T.}~\bibnamefont {Klein}}, \bibinfo {author} {\bibfnamefont {P.}~\bibnamefont {Lauria}}, \bibinfo {author} {\bibfnamefont {P.}~\bibnamefont {Lee}}, \bibinfo {author} {\bibfnamefont {D.}~\bibnamefont {Liefer}}, \bibinfo {author} {\bibfnamefont {S.~T.}\ \bibnamefont {Lu}}, \bibinfo {author} {\bibfnamefont {D.}~\bibnamefont {Lucchetti}}, \bibinfo {author} {\bibfnamefont {C.}~\bibnamefont {Lytle}}, \bibinfo {author} {\bibfnamefont {A.}~\bibnamefont {Malm}}, \bibinfo {author} {\bibfnamefont {M.}~\bibnamefont {Matheny}}, \bibinfo {author} {\bibfnamefont
				{B.}~\bibnamefont {Mathewson}}, \bibinfo {author} {\bibfnamefont {K.}~\bibnamefont {Mayer}}, \bibinfo {author} {\bibfnamefont {D.~B.}\ \bibnamefont {Miller}}, \bibinfo {author} {\bibfnamefont {M.}~\bibnamefont {Mills}}, \bibinfo {author} {\bibfnamefont {B.}~\bibnamefont {Neyenhuis}}, \bibinfo {author} {\bibfnamefont {L.}~\bibnamefont {Nugent}}, \bibinfo {author} {\bibfnamefont {S.}~\bibnamefont {Olson}}, \bibinfo {author} {\bibfnamefont {J.}~\bibnamefont {Parks}}, \bibinfo {author} {\bibfnamefont {G.~N.}\ \bibnamefont {Price}}, \bibinfo {author} {\bibfnamefont {Z.}~\bibnamefont {Price}}, \bibinfo {author} {\bibfnamefont {M.}~\bibnamefont {Pugh}}, \bibinfo {author} {\bibfnamefont {A.}~\bibnamefont {Ransford}}, \bibinfo {author} {\bibfnamefont {A.~P.}\ \bibnamefont {Reed}}, \bibinfo {author} {\bibfnamefont {C.}~\bibnamefont {Roman}}, \bibinfo {author} {\bibfnamefont {M.}~\bibnamefont {Rowe}}, \bibinfo {author} {\bibfnamefont {C.}~\bibnamefont {Ryan-Anderson}}, \bibinfo {author} {\bibfnamefont
				{S.}~\bibnamefont {Sanders}}, \bibinfo {author} {\bibfnamefont {J.}~\bibnamefont {Sedlacek}}, \bibinfo {author} {\bibfnamefont {P.}~\bibnamefont {Shevchuk}}, \bibinfo {author} {\bibfnamefont {P.}~\bibnamefont {Siegfried}}, \bibinfo {author} {\bibfnamefont {T.}~\bibnamefont {Skripka}}, \bibinfo {author} {\bibfnamefont {B.}~\bibnamefont {Spaun}}, \bibinfo {author} {\bibfnamefont {R.~T.}\ \bibnamefont {Sprenkle}}, \bibinfo {author} {\bibfnamefont {R.~P.}\ \bibnamefont {Stutz}}, \bibinfo {author} {\bibfnamefont {M.}~\bibnamefont {Swallows}}, \bibinfo {author} {\bibfnamefont {R.~I.}\ \bibnamefont {Tobey}}, \bibinfo {author} {\bibfnamefont {A.}~\bibnamefont {Tran}}, \bibinfo {author} {\bibfnamefont {T.}~\bibnamefont {Tran}}, \bibinfo {author} {\bibfnamefont {E.}~\bibnamefont {Vogt}}, \bibinfo {author} {\bibfnamefont {C.}~\bibnamefont {Volin}}, \bibinfo {author} {\bibfnamefont {J.}~\bibnamefont {Walker}}, \bibinfo {author} {\bibfnamefont {A.~M.}\ \bibnamefont {Zolot}},\ and\ \bibinfo {author} {\bibfnamefont
				{J.~M.}\ \bibnamefont {Pino}},\ }\bibfield  {title} {\bibinfo {title} {A {R}ace-{T}rack {T}rapped-{I}on {Q}uantum {P}rocessor},\ }\href {https://doi.org/10.1103/PhysRevX.13.041052} {\bibfield  {journal} {\bibinfo  {journal} {Phys. Rev. X}\ }\textbf {\bibinfo {volume} {13}},\ \bibinfo {pages} {041052} (\bibinfo {year} {2023})}\BibitemShut {NoStop}%
		\bibitem [{\citenamefont {Bravyi}\ \emph {et~al.}(2016)\citenamefont {Bravyi}, \citenamefont {Smith},\ and\ \citenamefont {Smolin}}]{BSS16}%
		\BibitemOpen
		\bibfield  {author} {\bibinfo {author} {\bibfnamefont {S.}~\bibnamefont {Bravyi}}, \bibinfo {author} {\bibfnamefont {G.}~\bibnamefont {Smith}},\ and\ \bibinfo {author} {\bibfnamefont {J.~A.}\ \bibnamefont {Smolin}},\ }\bibfield  {title} {\bibinfo {title} {Trading {C}lassical and {Q}uantum {C}omputational {R}esources},\ }\href {https://doi.org/10.1103/PhysRevX.6.021043} {\bibfield  {journal} {\bibinfo  {journal} {Phys. Rev. X}\ }\textbf {\bibinfo {volume} {6}},\ \bibinfo {pages} {021043} (\bibinfo {year} {2016})}\BibitemShut {NoStop}%
		\bibitem [{\citenamefont {Broadbent}\ and\ \citenamefont {Schaffner}(2016)}]{BS16}%
		\BibitemOpen
		\bibfield  {author} {\bibinfo {author} {\bibfnamefont {A.}~\bibnamefont {Broadbent}}\ and\ \bibinfo {author} {\bibfnamefont {C.}~\bibnamefont {Schaffner}},\ }\bibfield  {title} {\bibinfo {title} {Quantum cryptography beyond quantum key distribution},\ }\href {https://doi.org/10.1007/s10623-015-0157-4} {\bibfield  {journal} {\bibinfo  {journal} {Design. Code. Cryptogr.}\ }\textbf {\bibinfo {volume} {78}},\ \bibinfo {pages} {351} (\bibinfo {year} {2016})}\BibitemShut {NoStop}%
		\bibitem [{\citenamefont {Miyake}\ and\ \citenamefont {Miller}(2016)}]{MM16}%
		\BibitemOpen
		\bibfield  {author} {\bibinfo {author} {\bibfnamefont {A.}~\bibnamefont {Miyake}}\ and\ \bibinfo {author} {\bibfnamefont {J.}~\bibnamefont {Miller}},\ }\bibfield  {title} {\bibinfo {title} {Hierarchy of universal entanglement in {2D} measurement-based quantum computation},\ }\href {https://doi.org/10.1038/npjqi.2016.36} {\bibfield  {journal} {\bibinfo  {journal} {npj Quant. Inf.}\ }\textbf {\bibinfo {volume} {2}},\ \bibinfo {pages} {16036} (\bibinfo {year} {2016})}\BibitemShut {NoStop}%
		\bibitem [{\citenamefont {Gachechiladze}\ \emph {et~al.}(2019)\citenamefont {Gachechiladze}, \citenamefont {G{\"u}hne},\ and\ \citenamefont {Miyake}}]{GGM19}%
		\BibitemOpen
		\bibfield  {author} {\bibinfo {author} {\bibfnamefont {M.}~\bibnamefont {Gachechiladze}}, \bibinfo {author} {\bibfnamefont {O.}~\bibnamefont {G{\"u}hne}},\ and\ \bibinfo {author} {\bibfnamefont {A.}~\bibnamefont {Miyake}},\ }\bibfield  {title} {\bibinfo {title} {Changing the circuit-depth complexity of measurement-based quantum computation with hypergraph states},\ }\href {https://doi.org/10.1103/PhysRevA.99.052304} {\bibfield  {journal} {\bibinfo  {journal} {Phys. Rev. A}\ }\textbf {\bibinfo {volume} {99}},\ \bibinfo {pages} {052304} (\bibinfo {year} {2019})}\BibitemShut {NoStop}%
		\bibitem [{\citenamefont {Takeuchi}\ \emph {et~al.}(2019)\citenamefont {Takeuchi}, \citenamefont {Morimae},\ and\ \citenamefont {Hayashi}}]{TMH19}%
		\BibitemOpen
		\bibfield  {author} {\bibinfo {author} {\bibfnamefont {Y.}~\bibnamefont {Takeuchi}}, \bibinfo {author} {\bibfnamefont {T.}~\bibnamefont {Morimae}},\ and\ \bibinfo {author} {\bibfnamefont {M.}~\bibnamefont {Hayashi}},\ }\bibfield  {title} {\bibinfo {title} {Quantum computational universality of hypergraph states with {P}auli-{X} and {Z} basis measurements},\ }\href {https://doi.org/10.1038/s41598-019-49968-3} {\bibfield  {journal} {\bibinfo  {journal} {Sci. Repor.}\ }\textbf {\bibinfo {volume} {9}},\ \bibinfo {pages} {13585} (\bibinfo {year} {2019})}\BibitemShut {NoStop}%
		\bibitem [{\citenamefont {Raussendorf}\ \emph {et~al.}(2003)\citenamefont {Raussendorf}, \citenamefont {Browne},\ and\ \citenamefont {Briegel}}]{RBB03}%
		\BibitemOpen
		\bibfield  {author} {\bibinfo {author} {\bibfnamefont {R.}~\bibnamefont {Raussendorf}}, \bibinfo {author} {\bibfnamefont {D.~E.}\ \bibnamefont {Browne}},\ and\ \bibinfo {author} {\bibfnamefont {H.~J.}\ \bibnamefont {Briegel}},\ }\bibfield  {title} {\bibinfo {title} {Measurement-based quantum computation on cluster states},\ }\href {https://doi.org/10.1103/PhysRevA.68.022312} {\bibfield  {journal} {\bibinfo  {journal} {Phys. Rev. A}\ }\textbf {\bibinfo {volume} {68}},\ \bibinfo {pages} {022312} (\bibinfo {year} {2003})}\BibitemShut {NoStop}%
		\bibitem [{\citenamefont {Vijayan}\ \emph {et~al.}(2024)\citenamefont {Vijayan}, \citenamefont {Paler}, \citenamefont {Gavriel}, \citenamefont {Myers}, \citenamefont {Rohde},\ and\ \citenamefont {Devitt}}]{VPG+24}%
		\BibitemOpen
		\bibfield  {author} {\bibinfo {author} {\bibfnamefont {M.~K.}\ \bibnamefont {Vijayan}}, \bibinfo {author} {\bibfnamefont {A.}~\bibnamefont {Paler}}, \bibinfo {author} {\bibfnamefont {J.}~\bibnamefont {Gavriel}}, \bibinfo {author} {\bibfnamefont {C.~R.}\ \bibnamefont {Myers}}, \bibinfo {author} {\bibfnamefont {P.~P.}\ \bibnamefont {Rohde}},\ and\ \bibinfo {author} {\bibfnamefont {S.~J.}\ \bibnamefont {Devitt}},\ }\bibfield  {title} {\bibinfo {title} {Compilation of algorithm-specific graph states for quantum circuits},\ }\href {https://doi.org/10.1088/2058-9565/ad1f39} {\bibfield  {journal} {\bibinfo  {journal} {Quantum Sci. Technol.}\ }\textbf {\bibinfo {volume} {9}},\ \bibinfo {pages} {025005} (\bibinfo {year} {2024})}\BibitemShut {NoStop}%
		\bibitem [{\citenamefont {Kaldenbach}\ and\ \citenamefont {Heller}(2025)}]{KH25}%
		\BibitemOpen
		\bibfield  {author} {\bibinfo {author} {\bibfnamefont {T.~N.}\ \bibnamefont {Kaldenbach}}\ and\ \bibinfo {author} {\bibfnamefont {M.}~\bibnamefont {Heller}},\ }\bibfield  {title} {\bibinfo {title} {Mapping quantum circuits to shallow-depth measurement patterns based on graph states},\ }\href {https://doi.org/10.1088/2058-9565/ad802b} {\bibfield  {journal} {\bibinfo  {journal} {Quantum Sci. Technol.}\ }\textbf {\bibinfo {volume} {10}},\ \bibinfo {pages} {015010} (\bibinfo {year} {2025})}\BibitemShut {NoStop}%
		\bibitem [{\citenamefont {Peres}\ and\ \citenamefont {Galv{\~a}o}(2025)}]{PG25}%
		\BibitemOpen
		\bibfield  {author} {\bibinfo {author} {\bibfnamefont {F.~C.~R.}\ \bibnamefont {Peres}}\ and\ \bibinfo {author} {\bibfnamefont {E.~F.}\ \bibnamefont {Galv{\~a}o}},\ }\bibfield  {title} {\bibinfo {title} {Reducing depth and measurement weights in {P}auli-based computation},\ }\href {https://doi.org/10.1103/d3x5-cgky} {\bibfield  {journal} {\bibinfo  {journal} {Phys. Rev. A}\ }\textbf {\bibinfo {volume} {112}},\ \bibinfo {pages} {062604} (\bibinfo {year} {2025})}\BibitemShut {NoStop}%
		\bibitem [{\citenamefont {Fisher}\ \emph {et~al.}(2014)\citenamefont {Fisher}, \citenamefont {Broadbent}, \citenamefont {Shalm}, \citenamefont {Yan}, \citenamefont {Lavoie}, \citenamefont {Prevedel}, \citenamefont {Jennewein},\ and\ \citenamefont {Resch}}]{FBS+14}%
		\BibitemOpen
		\bibfield  {author} {\bibinfo {author} {\bibfnamefont {K.~A.~G.}\ \bibnamefont {Fisher}}, \bibinfo {author} {\bibfnamefont {A.}~\bibnamefont {Broadbent}}, \bibinfo {author} {\bibfnamefont {L.~K.}\ \bibnamefont {Shalm}}, \bibinfo {author} {\bibfnamefont {Z.}~\bibnamefont {Yan}}, \bibinfo {author} {\bibfnamefont {J.}~\bibnamefont {Lavoie}}, \bibinfo {author} {\bibfnamefont {R.}~\bibnamefont {Prevedel}}, \bibinfo {author} {\bibfnamefont {T.}~\bibnamefont {Jennewein}},\ and\ \bibinfo {author} {\bibfnamefont {K.~J.}\ \bibnamefont {Resch}},\ }\bibfield  {title} {\bibinfo {title} {Quantum computing on encrypted data},\ }\href {https://doi.org/10.1038/ncomms4074} {\bibfield  {journal} {\bibinfo  {journal} {Nat. Comm.}\ }\textbf {\bibinfo {volume} {5}},\ \bibinfo {pages} {3074} (\bibinfo {year} {2014})}\BibitemShut {NoStop}%
		\bibitem [{\citenamefont {Peres}\ and\ \citenamefont {Galv{\~{a}}o}(2023)}]{PG23}%
		\BibitemOpen
		\bibfield  {author} {\bibinfo {author} {\bibfnamefont {F.~C.~R.}\ \bibnamefont {Peres}}\ and\ \bibinfo {author} {\bibfnamefont {E.~F.}\ \bibnamefont {Galv{\~{a}}o}},\ }\bibfield  {title} {\bibinfo {title} {Quantum circuit compilation and hybrid computation using {P}auli-based computation},\ }\href {https://doi.org/10.22331/q-2023-10-03-1126} {\bibfield  {journal} {\bibinfo  {journal} {Quant.}\ }\textbf {\bibinfo {volume} {7}},\ \bibinfo {pages} {1126} (\bibinfo {year} {2023})}\BibitemShut {NoStop}%
		\bibitem [{\citenamefont {Lovsted}(2026)}]{Lov26}%
		\BibitemOpen
		\bibfield  {author} {\bibinfo {author} {\bibfnamefont {D.}~\bibnamefont {Lovsted}},\ }\emph {\bibinfo {title} {Resource-Efficient Secure Delegation via Pauli-Based Computation}},\ \href {https://doi.org/10.20381/ruor-32161} {Master's thesis},\ \bibinfo  {school} {University of Ottawa} (\bibinfo {year} {2026})\BibitemShut {NoStop}%
		\bibitem [{\citenamefont {Bremner}\ \emph {et~al.}(2016)\citenamefont {Bremner}, \citenamefont {Montanaro},\ and\ \citenamefont {Shepherd}}]{BMS16average}%
		\BibitemOpen
		\bibfield  {author} {\bibinfo {author} {\bibfnamefont {M.~J.}\ \bibnamefont {Bremner}}, \bibinfo {author} {\bibfnamefont {A.}~\bibnamefont {Montanaro}},\ and\ \bibinfo {author} {\bibfnamefont {D.~J.}\ \bibnamefont {Shepherd}},\ }\bibfield  {title} {\bibinfo {title} {{Average-Case Complexity Versus Approximate Simulation of Commuting Quantum Computations}},\ }\href {https://doi.org/10.1103/PhysRevLett.117.080501} {\bibfield  {journal} {\bibinfo  {journal} {Phys. Rev. Lett.}\ }\textbf {\bibinfo {volume} {117}},\ \bibinfo {pages} {080501} (\bibinfo {year} {2016})}\BibitemShut {NoStop}%
		\bibitem [{\citenamefont {Bremner}\ \emph {et~al.}(2017)\citenamefont {Bremner}, \citenamefont {Montanaro},\ and\ \citenamefont {Shepherd}}]{BMS17}%
		\BibitemOpen
		\bibfield  {author} {\bibinfo {author} {\bibfnamefont {M.~J.}\ \bibnamefont {Bremner}}, \bibinfo {author} {\bibfnamefont {A.}~\bibnamefont {Montanaro}},\ and\ \bibinfo {author} {\bibfnamefont {D.~J.}\ \bibnamefont {Shepherd}},\ }\bibfield  {title} {\bibinfo {title} {Achieving quantum supremacy with sparse and noisy commuting quantum computations},\ }\href {https://doi.org/10.22331/q-2017-04-25-8} {\bibfield  {journal} {\bibinfo  {journal} {Quant.}\ }\textbf {\bibinfo {volume} {1}},\ \bibinfo {pages} {8} (\bibinfo {year} {2017})}\BibitemShut {NoStop}%
		\bibitem [{\citenamefont {Codsi}\ and\ \citenamefont {van~de Wetering}(2025)}]{CvdW25}%
		\BibitemOpen
		\bibfield  {author} {\bibinfo {author} {\bibfnamefont {J.}~\bibnamefont {Codsi}}\ and\ \bibinfo {author} {\bibfnamefont {J.}~\bibnamefont {van~de Wetering}},\ }\bibfield  {title} {\bibinfo {title} {Classically simulating intermediate-scale instantaneous quantum polynomial circuits through a random graph approach},\ }\href {https://doi.org/10.1103/PhysRevA.111.012422} {\bibfield  {journal} {\bibinfo  {journal} {Phys. Rev. A}\ }\textbf {\bibinfo {volume} {111}},\ \bibinfo {pages} {012422} (\bibinfo {year} {2025})}\BibitemShut {NoStop}%
		\bibitem [{\citenamefont {Kissinger}\ and\ \citenamefont {van~de Wetering}(2019)}]{KvdW19}%
		\BibitemOpen
		\bibfield  {author} {\bibinfo {author} {\bibfnamefont {A.}~\bibnamefont {Kissinger}}\ and\ \bibinfo {author} {\bibfnamefont {J.}~\bibnamefont {van~de Wetering}},\ }\bibfield  {title} {\bibinfo {title} {Universal {MBQC} with generalised parity-phase interactions and {P}auli measurements},\ }\href {https://doi.org/10.22331/q-2019-04-26-134} {\bibfield  {journal} {\bibinfo  {journal} {Quant.}\ }\textbf {\bibinfo {volume} {3}},\ \bibinfo {pages} {134} (\bibinfo {year} {2019})}\BibitemShut {NoStop}%
		\bibitem [{\citenamefont {Ferguson}\ \emph {et~al.}(2021)\citenamefont {Ferguson}, \citenamefont {Dellantonio}, \citenamefont {Balushi}, \citenamefont {Jansen}, \citenamefont {D\"ur},\ and\ \citenamefont {Muschik}}]{FDB+21}%
		\BibitemOpen
		\bibfield  {author} {\bibinfo {author} {\bibfnamefont {R.~R.}\ \bibnamefont {Ferguson}}, \bibinfo {author} {\bibfnamefont {L.}~\bibnamefont {Dellantonio}}, \bibinfo {author} {\bibfnamefont {A.~A.}\ \bibnamefont {Balushi}}, \bibinfo {author} {\bibfnamefont {K.}~\bibnamefont {Jansen}}, \bibinfo {author} {\bibfnamefont {W.}~\bibnamefont {D\"ur}},\ and\ \bibinfo {author} {\bibfnamefont {C.~A.}\ \bibnamefont {Muschik}},\ }\bibfield  {title} {\bibinfo {title} {Measurement-{B}ased {V}ariational {Q}uantum {E}igensolver},\ }\href {https://doi.org/10.1103/PhysRevLett.126.220501} {\bibfield  {journal} {\bibinfo  {journal} {Phys. Rev. Lett.}\ }\textbf {\bibinfo {volume} {126}},\ \bibinfo {pages} {220501} (\bibinfo {year} {2021})}\BibitemShut {NoStop}%
		\bibitem [{\citenamefont {Lindell}(2017)}]{L17}%
		\BibitemOpen
		\bibfield  {author} {\bibinfo {author} {\bibfnamefont {Y.}~\bibnamefont {Lindell}},\ }\bibfield  {title} {\bibinfo {title} {How to simulate it–a tutorial on the simulation proof technique},\ }in\ \href@noop {} {\emph {\bibinfo {booktitle} {Tutorials on the Foundations of Cryptography: Dedicated to Oded Goldreich}}}\ (\bibinfo {year} {2017})\ pp.\ \bibinfo {pages} {277--346}\BibitemShut {NoStop}%
	\end{thebibliography}
\end{document}